%% file: main.tex
\documentclass[11pt]{article}

\usepackage{fullpage}
\usepackage[utf8]{inputenc} % allow utf-8 input
\usepackage[T1]{fontenc}    % use 8-bit T1 fonts
\usepackage{hyperref}       % hyperlinks
\usepackage{url}            % simple URL typesetting
\usepackage{booktabs}       % professional-quality tables
\usepackage{amsfonts}       % blackboard math symbols
\usepackage{nicefrac}       % compact symbols for 1/2, etc.
\usepackage{microtype}      % microtypography
\usepackage[english]{babel}
\usepackage{enumitem,bbm,verbatim,amscd,float,mathtools}
\usepackage[normalem]{ulem}
\usepackage{amsthm}
\usepackage{amsmath, amssymb, graphicx, url, algorithm}
\usepackage[noend]{algpseudocode}
\usepackage{times}
\input{glodef}
\input{locdef}

\title{Estimating~Entropy~of~Distributions~in~Constant~Space}

\author{
  Jayadev Acharya\\
  Cornell University \\
  \tt{acharya@cornell.edu}\\
  \and
  Sourbh Bhadane \\
  Cornell University\\
  \tt{snb62@cornell.edu} \\
  \and 
  Piotr Indyk\\
  Massachusetts Institute of Technology\\
  \tt{indyk@mit.edu} \\
  \and
  Ziteng Sun\\
  Cornell University\\
  \tt{zs335@cornell.edu}\\
}

\begin{document}

\maketitle

\begin{abstract}
We consider the task of estimating the entropy of $k$-ary distributions from samples in the streaming model, where space is limited. Our main contribution is an algorithm that requires $O\left(\frac{k \log (1/\varepsilon)^2}{\varepsilon^3}\right)$ samples and a constant $O(1)$ memory words of space and outputs a $\pm\varepsilon$ estimate of $H(p)$. Without space limitations, the sample complexity has been established as $S(k,\varepsilon)=\Theta\left(\frac k{\varepsilon\log k}+\frac{\log^2 k}{\varepsilon^2}\right)$, which is sub-linear in the domain size $k$, and the  current algorithms that achieve optimal sample complexity also require nearly-linear space in $k$. 

Our algorithm partitions $[0,1]$ into intervals and estimates the entropy contribution of probability values in each interval. The intervals are designed to trade off the bias and variance of these estimates. 
{\let\thefootnote\relax\footnote{{This work is supported by NSF-CCF-1657471. This research started with the support of MIT-Shell Energy Research Fellowship to JA and PI, while JA was at MIT. PI was supported by NSF TRIPODS award \#1740751 and a Simons Investigator Award.}}}
%Distribution property estimation and testing with limited memory is a largely unexplored research area.
\end{abstract}
\section{Introduction}
\input{intro.tex}

\subsection{Problem Formulation}
\input{formulation.tex}

\subsection{Prior Work}
\input{related-work.tex}

\input{preliminaries}
\input{one-interval.tex}

\section{Interval-based Algorithms}\label{sec:interval}
\input{two-intervals.tex}
\input{general-algorithm.tex}

\section{Open Problems}
There are several questions that arise from our work. While our algorithms require only a constant memory words of space, they require a $\log k$ multiplicative factor more samples (as a function of $k$) than the optimal sample complexity (in~\eqref{eqn:optimal}).
\begin{itemize}
\item 
Does there exist an algorithm for entropy estimation that has the optimal sample complexity and space requirement that is at most $\textrm{poly}(\log k)$?	
\end{itemize}
We are unaware of any implementation that requires sub-linear space in $k$. A simpler goal could be to design a strictly sublinear-space (space requirement $k^\alpha$ for some $\alpha<1$) sample-optimal algorithm. At the same time, there might not exist an algorithm with a small sample complexity. This leads to the following complementary question. 
\begin{itemize}
\item 
Prove a lower bound on the space requirement of a sample-optimal algorithm for entropy estimation. 	
\end{itemize}

Beyond these, obtaining sample-space trade-offs for distribution testing, and property estimation tasks is an exciting future direction. 
%There are no known implementations of the sample-optimal entropy estimation algorithms with space that is even sub-linear in $k$. Perhaps the most interesting, and challenging are proving lower bounds for space requirements for estimating entropy of discrete distributions. Is it possible to design a sample-optimal algorithm with $poly(\log k)$ space? %We hope our work encourages research into other distribution testing and property estimation problems in the streaming model. 

\bibliographystyle{unsrt}
\bibliography{masterref}
%\clearpage

\appendix  
%\input{app-preliminaries}
%\input{app-memory-requirements.tex}
\input{app-proofs-simple.tex}

\section{Proofs from Section~\ref{sec:simplealgo}}\label{sec:mdlpfs}
\section{Two interval Algorithm Proofs}\label{sec:twointproofs}

\input{two-intervalspf.tex}

\section{General Interval Algorithm} \label{pf:thm:general}
%\subsection{General Interval Algorithms}
%\label{app:general-algorithms}
%\input{app-general-algorithms}

\subsection{Unclipped Bias Bound} \label{sec:bias}
\input{general-bias.tex}
\subsection{Clipping Error Bound}\label{sec:clip}

\input{general-clipping.tex}
\subsection{Concentration Bound} \label{sec:var}
\input{general-variance.tex}

\end{document}

%% file: glodef.tex
\newtheorem{Theorem*}{Theorem}

\newtheorem{Claim*}[Theorem]{Claim}

\newtheorem{CounterExample*}{$\overline{\hbox{\bf Example}}$}

\newtheorem{Example*}[Theorem]{Example}

\newtheorem{Intuition*}[Theorem]{Intuition}
\newtheorem{Joke*}[Theorem]{Joke}

\newtheorem{Lemma*}[Theorem]{Lemma}
\newtheorem{Open problem}[Theorem]{Open problem}

\newtheorem{Question*}[Theorem]{Question}

\newtheorem{theorem}{Theorem}

\newtheorem*{corollary*}{Corollary}

\newtheorem{lemma}[theorem]{Lemma}
\newtheorem*{lemma*}{Lemma}

\def \bSubexa    {\begin{subexa}}

\usepackage[colorinlistoftodos,textsize=scriptsize]{todonotes}

\newcommand{\ignore}[1]{}

\newcommand{\EE}{\mathbb{E}}

\newcommand{\NN}{\mathbb{N}}

\newcommand{\RR}{\mathbb{R}}

\def \cA     {{\cal A}}

\def \cE     {{\cal E}}

\def \cI     {{\cal I}}

\def \cX     {{\cal X}}

\def \Paren#1{{\left({#1}\right)}}

\def\ignore#1{}

\newcommand{\bi}{\begin{itemize}}
\newcommand{\ei}{\end{itemize}}
\def\orpro{\mathop{\mathchoice
   {\vee\kern-.49em\raise.7ex\hbox{$\cdot$}\kern.4em}
   {\vee\kern-.45em\raise.63ex\hbox{$\cdot$}\kern.2em}
   {\vee\kern-.4em\raise.3ex\hbox{$\cdot$}\kern.1em}
   {\vee\kern-.35em\raise2.2ex\hbox{$\cdot$}\kern.1em}}\limits}

\def\andpro{\mathop{\mathchoice
 {\wedge\kern-.46em\lower.69ex\hbox{$\cdot$}\kern.3em}
 {\wedge\kern-.46em\lower.58ex\hbox{$\cdot$}\kern.25em}
 {\wedge\kern-.38em\lower.5ex\hbox{$\cdot$}\kern.1em}
 {\wedge\kern-.3em\lower.5ex\hbox{$\cdot$}\kern.1em}}\limits}

\def\simge{\mathrel{%
   \rlap{\raise 0.511ex \hbox{$>$}}{\lower 0.511ex \hbox{$\sim$}}}}

\def\simle{\mathrel{
   \rlap{\raise 0.511ex \hbox{$<$}}{\lower 0.511ex \hbox{$\sim$}}}}

\providecommand{\email}[1]{\href{mailto:#1}{\nolinkurl{#1}\xspace}}

\newcommand{\probof}[1]{\Pr\Paren{#1}}

\newcommand{\expectation}[1]{\EE\left[#1\right]}

\newcommand{\abs}[1]{\left\lvert #1 \right\rvert}

%% file: locdef.tex
\newcommand{\dP}{p}

\newcommand{\smb}{x}

\newcommand{\Xon}{X^\ns}

\theoremstyle{plain}
\newcommand{\Ent}[1]{H\left(#1\right)}
\newcommand{\Expc}[1]{\ensuremath{\mathbb{E}\left[#1\right]}}
\newcommand{\range}[2][]{\ifthenelse{\equal{#1}{}}{\left[#2\right]}{\left[#1,#2 \right]}}
\newcommand{\indctr}[1]{\mathbbm{1}\left\lbrace #1 \right \rbrace}

\allowdisplaybreaks

\def\withcolors{0}
\def\withnotes{0}

\ifnum\withcolors=1
  \newcommand{\new}[1]{{\color{red} {#1}}} % new
  \newcommand{\newer}[1]{{\color{blue} {#1}}} % even newer
  \newcommand{\acolor}[1]{{\color{purple}#1}} % Jayadev
  \newcommand{\bcolor}[1]{{\color{RubineRed}#1}} % Sourbh
  \newcommand{\scolor}[1]{{\color{orange}#1}} % Ziteng
  
\else
  \newcommand{\new}[1]{{{#1}}}
  \newcommand{\newer}[1]{{{#1}}}

  \newcommand{\acolor}[1]{{#1}}
  \newcommand{\bcolor}[1]{{#1}}
  
  \newcommand{\scolor}[1]{{#1}}
\fi

\ifnum\withnotes=1
  \newcommand{\anote}[1]{\par\acolor{\textbf{JA: }\sf #1}} % 
  \newcommand{\bnote}[1]{\par\bcolor{\textbf{SB: }\sf #1}} % 

\else
  \newcommand{\anote}[1]{}
  \newcommand{\bnote}[1]{}

\fi
\newcommand{\absz}{k}

\newcommand{\nsmp}{n}

\newcommand{\ent}{H}

\newcommand{\entdP}{\ent(\dP)}
\newcommand{\rentaldP}{\ent_\alpha(\dP)}

\newcommand{\renprm}{\alpha}

\newcommand{\eps}{\varepsilon}
\newcommand{\hhat}{\widehat{H}}

\newcommand{\logjab}[2]{\log^{(#1)}(#2)}
\newcommand{\logjabk}[1]{\logjab{#1}{\absz}}
\newcommand{\iter}{T}

\newcommand{\ns}{\nsmp}

\newcommand{\p}{\dP}

\newcommand{\SamC}[3]{\mathbb{S}\left(#1,#2,#3\right)}
\newcommand{\vSmp}{R}
\newcommand{\bSmp}{N}

\newcommand{\estII}{\widehat{H}_{II}}
\newcommand{\estIIs}{\widehat{H}_{II}^*}
\newcommand{\estis}{\widehat{H}_{\mathcal{I}}^*}
\newcommand{\esti}{\widehat{H}_{\mathcal{I}}}

\newcommand{\prt}{a}
\newcommand{\bset}[1]{\left \lbrace \bSmp_i \right \rbrace_{i = 1}^{#1} }
\newcommand{\vset}[1]{\left \lbrace \vSmp_i \right \rbrace_{i = 1}^{#1}}

%% file: intro.tex
%\subsection{Streaming Algorithms}
%\vspace{-10pt}
\noindent\textbf{Streaming Algorithms.}
Algorithms that require a limited memory/space/storage\footnote{We use space, storage, and memory interchangeably.} have garnered great interest over the last two decades, and are known as \emph{streaming algorithms}. Initiated by~\cite{MunroP80,FlajoletM85}, this setting became mainstream with the seminal work of~\cite{AlonMS96}. Streaming algorithms are particularly useful in handling massive datasets that cannot to be stored in the memory of the system. It is also applicable in networks where data is naturally generated sequentially and the data rates are higher than the capabilities of storing them, e.g., on a router. 

The literature on streaming algorithms is large, and many problems have been studied in this model. With roots in computer science, most of this literature considers the worst case model, where it is assumed that the input $\Xon:=X_1, \ldots, X_n$ is an arbitrary sequence over a domain of size $k$ (e.g., over $[k]:=\{1, \ldots, k\}$). The set-up is as follows:

\textit{Given a system with limited memory that can make a few (usually just one) passes over the input $\Xon$, estimate $f(\Xon)$ for some function $f$ of interest. The primary objective is solving the task with the least memory, which is called the \emph{space complexity}.}

In this streaming setting, some well studied problems include estimation of frequency moments of the data stream~\cite{AlonMS96, Indyk00, IndykW05}, estimation of Shannon and R\'enyi entropy of the empirical distribution of the data stream~\cite{LallSOXZ06, ChakrabartiDM06, ChakrabartiCM10, HarveyNO08, GuhaMV09}, estimation of heavy hitters~\cite{CharikarCF02, CormodeM05, MetwallyAE05, BhattacharyyaDW16}, and estimation of distinct elements~\cite{BarJRST02, KaneNW10}. While these consider worst case input, there has also been work on random order streams, where one still considers a worst case data stream $\Xon$, but feeds a random permutation $X_{\sigma(1)},\ldots, X_{\sigma(n)}$ of $\Xon$ as input to the algorithm~\cite{GuhaMV09, GuhaM09, ChakrabartiJP08}.
%\subsection

\medskip
\noindent\textbf{Statistical Estimation.}
Inferring properties of the underlying distribution given sample access is called statistical estimation. A typical set-up is as follows: 

\textit{Given independent samples $X_1, \ldots, X_n$  from an unknown distribution $\p$, the objective is to estimate a property $f(\p)$ using the fewest samples, called the \emph{sample complexity}.}

Distribution property estimation literature most related to our work include entropy estimation~\cite{Miller55, Paninski03, ValiantV11a, JiaoVHW15, WuY16, AcharyaDOS17, NIPS2018_8099, CharikarSS19}, support size estimation~\cite{ValiantV11a, WuY16, hao2019broad}, R\'enyi entropy estimation~\cite{AcharyaOST17, ObremskiS17, FukuchiS17}, support coverage estimation~\cite{EfronT76, OrlitskySW16}, and divergence estimation~\cite{HanJW16, Acharya18}. %In these tasks, the optimal sample complexity is \emph{sub-linear} in $k$, the domain size of the distribution. 

%\subsection

\medskip
\noindent\textbf{Streaming Algorithms for Statistical Estimation.}
While space complexity of streaming algorithms, and sample complexity of statistical estimation have both received great attention, statistical estimation under memory constraints has received relatively little attention. Interestingly, almost half a century ago, Cover and Hellman~\cite{HellmanC70, Cover69} studied statistical hypothesis testing with limited memory, and~\cite{LeightonR86} studied estimating the bias of a coin using a finite state machine. However, until recently, there are few works on learning with memory constraints. There has been a recent interest in space-sample trade-offs in statistical estimation~\cite{GuhaM07, ChienLM10, DaganS18, CrouchMW16, SteinhardtVW16, Raz16, MoshkovitzM17, JainT18}. Within these,~\cite{CrouchMW16} is the closest to our paper. They consider estimating the integer moments of distributions, which is equivalent to estimating R\'enyi entropy of integer orders under memory constraints. They present natural algorithms for the problem, and perhaps more interestingly, prove non-trivial lower bounds on the space complexity of this task. Very recently, a remarkable work of~\cite{DiakonikolasGKR19} obtained memory sample trade-offs for testing discrete distributions, which are tight in a some parameter regime. 

We initiate the study of distribution entropy estimation with space limitations, with the goal of understanding the space-sample trade-offs.

%% file: formulation.tex
Let $\Delta_k$ be the set of all $k$-ary discrete distributions over $\cX=[k]:=\{0,1,\ldots, k-1\}$.
%\new{which is commonly used in both streaming and statistical estimation literature.} 
The Shannon entropy of $\p\in\Delta_k$ is $ \Ent{p} := -\sum_{x\in[k]} p\left(x\right) \log\left({p\left(x\right)}\right).$
%\vspace{-5pt}
% \begin{equation}
% \Ent{p} := -\sum\limits_{x\in[k]} p\left(x\right) \log\left({p\left(x\right)}\right).
% \end{equation}
Entropy is a fundamental measure of randomness and  a central quantity in information theory and communications. Entropy estimation is a key primitive  for feature selection in various machine learning applications.

%\subsubsection{Entropy Estimation Problem.} 
Given independent samples $\Xon:=X_1, \ldots, X_n$ from an unknown $\p\in\Delta_k$, an entropy estimator is a possibly randomized mapping $\hhat:[k]^n\to \RR$. Given $\eps>0$, $\delta >0$, $\hhat$ is an $(\eps,\delta)$ entropy estimator if
%\vspace{-10pt}
\begin{align}
\label{eqn:estimator}
{\sup_{\p\in\Delta_k}}{\rm Pr}_{\Xon \stackrel{i.i.d.}{\sim} \p}\Paren{|\hhat(\Xon)-\entdP|>\eps}<\delta.
\end{align}

\smallskip
\noindent\textbf{Sample Complexity.}
The \emph{sample complexity} $S(H, k, \eps, \delta)$ is the least $n$ for which a $\hhat$ satisfying~\eqref{eqn:estimator} exists. Throughout this paper, we assume a constant error probability, say $\delta=1/3$,\footnote{For smaller $\delta$'s, we can apply median trick with an extra factor of $\log(1/\delta)$ samples.} and exclusively study entropy estimation. We therefore denote $S(H, k, \eps, 1/3)$  by $S(k, \eps)$. 

\smallskip
\noindent\textbf{Memory Model and Space Complexity.} The basic unit of our storage model is a \emph{word}, which consists of $\log k+\log (1/\eps)$ bits. This choice of storage model is motivated by the fact that at least $\log (1/\eps)$ bits are needed for a precision of $\pm \eps$, and $\log k$ bits are needed to store a symbol in $[k]$. The \emph{space complexity} of an algorithm is the smallest space (in words) required for its implementation.

%\begin{definition}
%Given $k, \eps>0$ we say that $(m,n)$ is an achievable memory-sample pair, if there exists an algorithm that uses $m$ words of storage, and $n$ samples from $\p\in \Delta_k$, and outputs a $\pm\eps$ estimate of $H(\p)$. Let $\texttt{MS}(H, k,\eps)$ be the closure of the set of all achievable $(m,n)$. 
%\end{definition}

%\begin{definition} 
%The sample complexity of $\widehat{H}$ is the minimum number of samples it uses to estimate $H\left(p\right)$ such that $\lvert \widehat{H} - H\left(p\right) \rvert \leq \epsilon$ with probability atleast $1-\delta$. 
%\begin{equation}
%\SamC{k}{\eps}{\delta} = \min \left \lbrace n : \Pr\left(\lvert \widehat{H} - H\left(p\right) \rvert \leq \eps \right) \geq 1-\delta \right \rbrace.
%\end{equation} 
%\end{definition}

%% file: related-work.tex
\textbf{Distribution Entropy Estimation.}
Entropy estimation from samples has a long history~\cite{Miller55, Basharin59, AntosK01}. The most popular method is empirical plug-in estimation that outputs the entropy of the empirical distribution of the samples. Its sample complexity~\cite{AntosK01, Paninski03} is
%\vspace{-5pt}
%\begin{align}
%\label{eqn:empirical}
%S^{e}(k,\eps)=\Theta\left(\frac\absz\eps+\frac{\log^2 k}{\eps^2}\right).
%\end{align}
\begin{align}
\label{eqn:empirical}
S^{e}(k,\eps)=\Theta\left(\frac{\absz}{\eps}+\frac{\log^2 k}{\eps^2}\right).
%S^{e}(k,\eps)=\Theta\left({\absz}/{\eps}+({\log^2 k})/{\eps^2}\right).
\end{align}
\noindent Paninski \cite{Paninski04b} showed that there exists an  estimator with sub-linear sample complexity in $\absz$. A recent line of work~\cite{ValiantV11a, WuY16, JiaoVHW15}  has characterized the optimal sample complexity  as
%\vspace{-5pt}
\begin{align}
%	S(k,\eps) = \Theta\left(\frac{\absz}{\eps\log \absz}+\frac{\log^2 k}{\eps^2}\right).
		S(k,\eps) = \Theta\left(\frac{\absz}{\eps\log \absz}+\frac{\log^2 k}{\eps^2}\right).\label{eqn:optimal}
\end{align}
%\begin{align}
%S(k,\eps) = \Theta\left(\frac\absz{\eps\log \absz}+\frac{\log^2 k}{\eps^2}\right).\label{eqn:optimal}
%\end{align}
Note that while the empirical estimator has a linear sample complexity in the domain size $k$, the optimal sample complexity is sub-linear.

\smallskip
\noindent\textbf{Estimating Entropy of Streams.} 
There is significant work on estimating entropy of the stream with limited memory. Here, there are no distributional assumptions on the input stream $\Xon$, and the goal is to estimate $H(\Xon)$, the entropy of the empirical distribution of $\Xon$.~\cite{LallSOXZ06, BhuvanagiriG06, GuhaMV09, HarveyNO08, ChakrabartiCM10} consider multiplicative entropy estimation. These algorithms can be modified to additive entropy estimation by noting that $(1\pm\eps/\log {\min\{n,k\}})$ multiplicative estimation yields a $\pm\eps$ additive estimation. With this,~\cite{ChakrabartiCM10, GuhaMV09} give an algorithm requiring $O(\frac{\log^3 \nsmp}{\eps^2})$ words of space for $\pm\eps$ estimate of $H(\Xon)$.~\cite{HarveyNO08} give an algorithm using $O(\frac{\log^2 \nsmp\cdot \log \log \nsmp}{\eps^2})$ words of space. A space lower bound of $\Omega(1/\eps^2)$ words was also proved in~\cite{ChakrabartiCM10} for the worst-case setting.

Another widely used notion of entropy is R\'enyi entropy~\cite{Renyi61}. The R\'enyi entropy of $\dP$ of order $\alpha>0$ is 
$\rentaldP:=\log(\sum_{\smb}\p(x)^{\renprm})/({1-\renprm})$.
\cite{GoldreichR00, BarKS01, AcharyaOST17} show that the sample complexity  of estimating $\rentaldP$ is $\Theta({\absz^{1-1/\alpha}}/{\eps^2})$ for $\alpha\in\{2,3,\ldots\}$.~\cite{CrouchMW16} studies the problem of estimating the collision probability, which can be seen as estimating $\rentaldP$ for $\renprm=2$, under memory constraints. They propose an algorithm  with sample complexity $\nsmp$ and the memory $M$ satisfying  $\nsmp\cdot M\ge \Omega(\absz)$, when $\nsmp$ is at least $O(\absz^{1-1/\renprm})$. They also provide some (non-tight) lower bounds on the memory requirements.

%This leads to an algorithm with sample complexity $S^{e}(k,\eps)$, and space complexity $O(\log^3 (S^{e}(k,\eps))/\eps^2)$ words to estimate $H(\p)$. To the best of our knowledge, the best space-sample trade-off is obtained by using. Furthermore, in the streaming set-up, a space lower bound of $\Omega(1/\eps^2)$ was proved in~\cite{ChakrabartiCM10}. 
%\vspace{-5pt}
\subsection{Our Results and Techniques}

Our goal is to design streaming algorithms for estimation of $H(\p)$ from a data stream of samples $\Xon\sim\p$, with as little space as possible. Our motivating question is:

\smallskip
\qquad\qquad\qquad\emph{What is the space-sample trade-off of entropy estimation over $\Delta_k$?}

%The optimal sample complexity is given in~\eqref{eqn:optimal}. However, the sample-optimal implementations in~\cite{ValiantV11a, WuY16, JiaoVHW15} require nearly linear (in $k$) words space. The straight-forward implementation of these schemes require nearly linear space complexity in $S(k,\eps)$. 
\smallskip
The optimal sample complexity is given in~\eqref{eqn:optimal}. However, straight-forward implementations of sample-optimal schemes in~\cite{ValiantV11a, WuY16, JiaoVHW15} require nearly linear space complexity in $S(k,\eps)$, which is nearly linear (in $k$) words of space. Note by~\eqref{eqn:empirical} that when the number of samples is at least $S^{e}(k,\eps)$, the empirical entropy $H(\Xon)$ is a $\pm\eps$ estimate of $H(\p)$. We can therefore use results from streaming literature to estimate the empirical entropy of $\Xon$ with $n= S^{e}(k,\eps)$ samples to within $\pm\eps$, and in doing so, obtain a $\pm2\eps$ estimate of $H(\p)$.
%\textit{If there is a streaming algorithm for estimating the entropy of a randomized data-stream of length $n= S^{e}(k,\eps)$ with a memory of $M_n$ words, then there is an algorithm to estimate the entropy of a distribution using $S^{e}(k,\eps)$ samples and $M_n$ words.}
In particular, the algorithm of~\cite{HarveyNO08} requires $S^{e}(k,\eps)$ samples, and with $O(\frac{\log^2 (k/\eps)\log \log (k/\eps)}{\eps^2})$ words of space, estimates $H(\p)$ to $\pm\eps$. 
%\todo{Check this part carefully.}

\new{Our main result is an algorithm whose space complexity is a constant number of words and whose sample complexity is linear in $k$ (same as empirical estimation as a function of $k$).}
\begin{theorem} \label{thm:main_perf}
There is an algorithm that requires
$O\Paren{\frac{\absz(\log (1/\eps))^2}{\eps^3}}$ samples
and 20 words of space and estimates $H(\p)$ to $\pm \eps$. 
%Here $\log^* k := \min_i \{\log^{(i)} k \le 1\}$. For $i=1$, $\log^{(i)} k = \log k$, and for $i > 1$, $\log^{(i)} k = \max\left\{\log (\log^{(i-1)} k), 1\right\}$. 
\end{theorem}
 
The results and the state of the art are given in Table~\ref{tab:results}. A few remarks are in order.

\noindent \textbf{Remark.}
 (1). Our algorithm can bypass the lower bound of $\Omega(1/\eps^2)$ for entropy estimation of data-streams since $\Xon$ is generated by a distribution and not the worst case data stream. 
(2). Consider the case when $\eps$ is a constant, say $\eps=1$. Then, the optimal sample complexity is $\Theta(\frac k{\log k})$ (from~\eqref{eqn:optimal}). However, all known implementations of the sample-optimal algorithms requires $\tilde \Omega (k)$ space. Streaming literature provides an algorithm with $O(k)$ samples and $\tilde O((\log k)^2)$ memory words. We provide an algorithm with $O(\absz)$ samples, and 20 memory words. Compared to the sample-optimal algorithms,  we have a $\log k$ blow-up in the sample complexity, but an exponential reduction in space. 

%\vspace{-10pt}
\begin{table}[htb]
\caption{Sample and space complexity for estimating $\entdP$.}
	\begin{center}
		\begin{tabular}{|c|c|c|c|c|c|c|}
			\hline
			Algorithm & Samples & Space (in words) \\ \hline
			\shortstack{Sample-Optimal~~\cite{ValiantV11a},\cite{ WuY16, JiaoVHW15} }& $\Theta\left(\frac{\absz}{\eps\log \absz}+\frac{\log^2 k}{\eps^2}\right)$ & $O\left(\frac{\absz}{\eps\log \absz}+\frac{\log^2 k}{\eps^2}\right)$ \\ \hline 
			\shortstack{Streaming  \cite{ChakrabartiCM10, HarveyNO08} }& $O\left(\frac{\absz}{\eps}+\frac{\log^2 k}{\eps^2}\right)$ & $O\Paren{{\log^2 (\frac{\absz}{\eps})\log \log (\frac\absz{\eps})}/{\eps^2}}$\\\hline
			Algorithm~\ref{alg:genint} & $O\Paren{\frac{\absz\Paren{\log (1/\eps)}^2}{\eps^3}}$ &  $20$\\ \hline
			%Probability multiset & $\sets{p_1\upto p_k}$ & & & & & optimal \\ \hline
		\end{tabular}
	\end{center}
	%\vspace{-20pt}
	\label{tab:results}
\end{table}
%\vspace{-10pt}

%
%\renewcommand{\arraystretch}{1.6}

%\renewcommand{\arraystretch}{1.0}
%
%
%\renewcommand{\arraystretch}{2.0}
%\subsection{Our Techniques}
\label{sec:techniques}
\noindent We now provide a high level description of our approach and techniques. We can write $H(\p)$ as 
 %\vspace{-10pt}
\begin{align}
\label{eqn:expansion-entropy}
H(\p) = \sum_x \p(x)\log \frac1{\p(x)} = \EE_{X\sim p}\left[ \log \frac1{\p(X)} \right].
\end{align}

\noindent\textbf{A Simple Method.} Based on this equation, we build layers of sophistication to a simple approach which requires small space.  Repeat for $R$ iterations:

\noindent\quad 1. Obtain a draw $X\sim p$.

\noindent\quad 2. Using constant memory words, over the next $N$ samples, estimate $\log (1/\p(X))$, and maintain a running average over the iterations.

%\vspace{-\topsep} 
%\begin{itemize}
%\setlength\itemsep{0em}
%\item Obtain a sample $X\sim p$. 
%\item
%Using constant memory, over the next $N$ samples, obtain an estimate of $\log (1/\p(X))$. 
%\end{itemize}
%\vspace{-\topsep}
We need $N$ to be large enough to obtain a \emph{good estimate} $\widehat p(X)$ of $\p(X)$ for the term inside the expectation in~\eqref{eqn:expansion-entropy}, and we need $R$ large enough for the empirical means of $\log (1/\widehat p(X))$ over $R$ iterations to converge to the true mean. The number of samples needed is $NR$. This approach is detailed in Algorithm~\ref{alg1} (in Section~\ref{sec:simplealgo}) and its performance is given in Theorem~\ref{thm:one-interval}. This approach requires $O(1)$ memory words, however the sample complexity is super-linear in $k$. 
%\begin{theorem}\label{thmnaive}
% Algorithm~\ref{alg1} uses $O\left(1\right)$ words of memory and its sample complexity is 
 %\begin{equation}
 %\SamC{k}{\eps}{\frac{1}{4}} = O\left(\frac{k \log^2 (k/\eps)}{\eps^3}\right).
 %\end{equation}\todo{Can we remove this theorem from here? }
 %\end{theorem}
 
\medskip 
\noindent\textbf{Intervals for Better Sample Complexity.} To improve the sample complexity, we partition $\left[ 0,1\right]$ into $T$ disjoint intervals (Algorithm~\ref{alg1} corresponds to $T=1$). In Lemma~\ref{lem:decom} we express $H(\p)$ as a sum of entropy-like expressions defined over probability values in these $T$ intervals. We will then estimate each of the terms separately with the approach stated above. We will show that the sample complexity as a function of $k$ drops down roughly as $k (\log^{(T)} k)^2$, where $\log^{(T)}$ is the $T$th iterated logarithm, while the space complexity is still constant memory words. %We flesh this idea out in Section~\ref{sec:two_interval} for $T=2$. The choice of intervals, and the sample complexity guarantees for $T=2$ (Theorem~\ref{thm:perf_two}) will clarify our arguments and how they extend to larger values of $T$.% \new{The general algortihm is described in Section~\ref{sec:general_interval}}.

The algorithm will essentially perform the simple approach above separately for probabilities within each interval. While simple to state, there are several bells and whistles needed to make this approach work. The essence is that when $p(X)$ is large, fewer samples are needed to estimate $p(X)$ (small $N$). However, if the intervals are chosen such that small probabilities are also contained in small intervals, the number of iterations $R$ needed for these intervals can be made small (the range of random variables in Hoeffding's inequality is smaller). Succinctly, the approach can be summarized as follows:
 
\textit{Fewer samples are needed to estimate the large probabilities, and fewer iterations are needed for convergence of estimates for small probabilities by choosing the intervals carefully.}

%% file: preliminaries.tex
\medskip
\noindent\textbf{Some Useful Tools.}
\label{sec:preliminaries}
We now state two concentration inequalities that we use throughout this paper. 
%We use bounds on expressions of Binomial random variables, and concentration inequalities. 
%
%%\vspace{-5pt}
%\begin{lemma}\label{binreciplemma}
%Let $X \sim \text{Bin} \left(m,r\right)$, then
%$
%\Expc{\frac{1}{X+1}} \leq \frac{1}{r\left(m+1\right)}.
%$%\end{equation*}
%\end{lemma}
%%\begin{proof}
%%\vspace{-\topsep}
%\begin{proof}
%	\begin{align*}
%	\Expc{\frac{1}{X+1}} = \frac1{m+1}\sum\limits_{l=0}^{m} \frac{m+1}{l+1} {{m}\choose{l}} r^l\left(1-r\right)^{m-l}= \frac{1-(1-r)^{m+1}}{r\left(m+1\right)}\leq \frac{1}{r\left(m+1\right)}.
%	\end{align*}
%\end{proof}

%	\begin{align*}
%\text{\textit{Proof.\quad}}\Expc{\frac{1}{X+1}} = \frac1{m+1}\sum\limits_{l=0}^{m} \frac{m+1}{l+1} {{m}\choose{l}} r^l\left(1-r\right)^{m-l}= \frac{1-(1-r)^{m+1}}{r\left(m+1\right)}\leq \frac{1}{r\left(m+1\right)}.\qquad\quad\ensuremath{\square}
%\end{align*}
%\end{proof}
%%
\begin{lemma}(\textbf{Hoeffding's Inequality})~\cite{Hoeffding63} \label{lem:hoeff}
Let $X_1, \ldots, X_m\in[a_i, b_i]$ be independent random variables. Let $X =  (X_1+\ldots+X_m)/m$, then 
$	\probof{\left \lvert X - \Expc{X} \right \rvert \geq t } \leq 2\exp\left(\frac{-2(mt)^2}{\sum_i\left(b_i-a_i\right)^2} \right).
$%	\end{equation*}
\end{lemma}

In some algorithms we consider, $m$ itself is a random variable. In those cases, we will use the following variant of Hoeffding's inequality, which  is proved in Section~\ref{app:simple}.

%\begin{lemma} (\textbf{Random Hoeffding's Inequality}) \label{lem:ranhoeff}
%	Let $M \sim \text{Bin}\left(m,p\right)$. Let $X_1, \ldots, X_m$ be independent random variables such that $X_i \in \left[a,b\right]$. Let $X = (\sum_{i=1}^{M}X_i)/M$. Then, for any $0<p\le 1$
%	\begin{equation}\label{eq:RH}
%	\Pr \left( \left \lvert X - \Expc{X} \right \rvert \geq \frac{t}{p} \right) \leq 2\exp\left(\frac{-mt^2}{p\left(b-a\right)^2} \right) + \exp\left(-\frac{mp}{8} \right).
%	\end{equation}
%\end{lemma}
%
%
%\noindent \textbf{Remark.}  Since $|X - \EE\left[X\right] | \leq b-a$, unless $p \geq \frac{t}{b-a}$, we have $\Pr \left( \left \lvert X - \Expc{X} \right \rvert \geq \frac{t}{p} \right) = 0$. This implies that $\exp\left(- \frac{mp}{8} \right) \leq \exp \left( \frac{-mt^2}{8p (b-a)^2}\right)$. Therefore up to constant factors, it is sufficient to bound the first term of \eqref{eq:RH}. Henceforth, when we use Lemma~\ref{lem:ranhoeff}, we will only bound the first term of \eqref{eq:RH}.
\new{
\begin{lemma} (\textbf{Random Hoeffding's Inequality}) \label{lem:ranhoeff}
	Let $M \sim \text{Bin}\left(m,p\right)$. Let $X_1, \ldots, X_M$ be independent random variables such that $X_i \in \left[a,b\right]$. Let $X = (X_1+\ldots+X_M)/M$. Then for any $0<p\le 1$,
	\begin{equation}\label{eq:RH}
	\Pr \left( \left \lvert X - \Expc{X} \right \rvert \geq \frac{t}{p} \right) \leq 3\exp\left(\frac{-mt^2}{8p\left(b-a\right)^2} \right).
	\end{equation}
\end{lemma}
%\vspace{-15pt}
}

\medskip
\noindent\textbf{Outline.} In Section~\ref{sec:simplealgo} we describe the simple approach and its performance in Theorem~\ref{thm:one-interval}. In Section~\ref{sec:two_interval}, Algorithm~\ref{alg:twoint} we show how the sample complexity in Theorem~\ref{thm:one-interval} can be reduced from $k\log^2 k$  to $k(\log \log k)^2$ in Theorem~\ref{thm:perf_two} by choosing two intervals ($T=2$). \new{The algorithm for general $T$ is described in Section~\ref{sec:general_interval}, and the performance of our main algorithm is given in Theorem~\ref{thm:main_perf}.}

%% file: one-interval.tex
%The key idea in the \naive algorithm is to estimate $\log\left(p \left(x\right) \right)$ independent of $p\left(x\right)

\section{A Building Block: Simple Algorithm with Constant Space}\label{sec:simplealgo}

%The idea behind the model algorithm is based on the following observation:
%	\[
%	H(p) = \EE_{X \sim p} [\log \frac{1}{p(X)}]
%\]
We propose a simple method (Algorithm~\ref{alg1}) with the following guarantee. 
%\vspace{-5pt}
\begin{theorem}
\label{thm:one-interval}
Let $\eps>0$. Algorithm~\ref{alg1} takes $O\left(\frac{k\log^2 \left(k/\eps\right)}{\eps^3}\right)$ samples from $\p\in\Delta_k$, uses at most 20 words of memory, and outputs $\bar H$, such that with probability at least $2/3$, $\abs{\bar H - H(\p)}<\eps$.
\end{theorem}

%Algorithm~\ref{alg1} elaborates the two steps described in the beginning of Section~\ref{sec:techniques} by specifying precise values of $N$ and $R$ needed to bound the bias and concentration. The basis of the approach is from~\eqref{eqn:expansion-entropy}. 

Based on ~\eqref{eqn:expansion-entropy}, each iteration of Algorithm~\ref{alg1} obtains a draw $X$ from $\p$ and estimates $\log(1/\p(X))$. To avoid assigning zero probability value to $\p(X)$, we do add-1 smoothing to our empirical estimate of $\p(X)$. The bias in our estimator can be controlled by the choice of $N$. %The final estimator is the empirical mean of $R$ independent estimates $\widehat H_1 \ldots \widehat H_R$.
\begin{algorithm}[ht]
	\caption{Entropy estimation with constant space: Simple Algorithm} \label{alg1}
	\begin{algorithmic}[1]
	\Require Accuracy parameter $\eps>0$, a data stream $X_1, X_2, \ldots\sim p$
		\State Set 
		%\vspace{-10pt}
%		\[
%		R \gets \frac{4\log^2(1+2k/\eps)}{\eps^2},\qquad \bSmp \gets \frac{2k}{\eps}, \qquad S \gets 0
%		\]
		\[
			R \gets 4\log^2(1+2k/\eps)/\eps^2,\qquad \bSmp \gets 2k/\eps, \qquad S \gets 0
		\]
		%\vspace{-20pt}
		\For{$t=1, \ldots, \vSmp $}
		\State Let $x\gets$ the next element in the data stream 
		\State $N_{x}\gets$ \# appearances of $x$ in the next $\bSmp$ symbols
		%\State $\widehat{H}_t = \log\left(\frac{\bSmp}{N_{x} + 1}	\right) $
		\State $\widehat{H}_t = \log\left({\bSmp}/({N_{x} + 1})	\right) $		 \label{entstepalg1}
		\State  $S = S + \widehat{H}_t$
		\EndFor
		\State $\bar{H} = S/ \vSmp$
	\end{algorithmic}
\end{algorithm}

%Algorithm~\ref{alg1}  We explicitly count these in .
%\medskip
%\noindent\textbf{Memory Requirements of Algorithm~\ref{alg1}.}
%\newer{We  reserve two words for $\bSmp$, $\vSmp$. We reserve one word each to store $S,x$ and each of the $N$ symbols. We reserve two words for $N_x$ and three words for the counters. Thus the algorithm uses less than 20 words of space.}
%\anote{Show that the algorithm needs at most 10 words of space?}
%Algorithm~\ref{alg1} only maintains a running sum at the end of each iteration can be implemented with 20 words of space (Section~\ref{sec:memory}). To bound the accuracy, note that $\bar H$ is the mean of $R$ i.i.d. random variables $\widehat H_1, \ldots, \widehat H_R$. We bound the bias and prove concentration of $\bar H$ using Lemma~\ref{lem:hoeff}.
%Algorithm~\ref{alg1} only maintains a running sum at the end of each iteration. Note that $\bSmp \le \frac{k^3}{\eps^3}, \vSmp \le \frac{k^2}{\eps^2}$. We reserve three words for $\bSmp$ and two words for $\vSmp$. Since $S \le k$ and it suffices to store $S$ up to accuracy $\eps^2$. So we use two words to store $S$. We reserve one word to store $x$ and two words to keep track of $N_x$ in each iteration. We reserve three words for the counters. Thus the algorithm uses less than 20 words of space.

\noindent \textbf{Memory Requirement.} Algorithm~\ref{alg1} only maintains a running sum at the end of each iteration. We can use two words in total to store $k$ and $\eps$. Since $\bSmp, \vSmp$ are program constants they are computed on the fly. We reserve one word to store $x$ and two words to keep track of $N_x$ in each iteration since $N_x \le N \le k^2/\eps^2$. We use three words to store the counter $t$ since $t \le R \le k^3/\eps^3$. We use two words each to store $S$ and $\hhat_t$ (and store the final entropy estimate in one of them). Thus the algorithm uses less than 20 words of space.

\noindent\textbf{Sample Complexity.} To bound the accuracy, note that $\bar H$ is the mean of $R$ i.i.d. random variables $\widehat H_1, \ldots, \widehat H_R$. We bound the bias and prove concentration of $\bar H$ using Lemma~\ref{lem:hoeff}.

%\medskip
\noindent\textit{Bias Bound.} Large $N$ in Algorithm~\ref{alg1} gives a better estimate of $\p(X)$, and small bias in entropy estimation.
\begin{lemma}(\textbf{Bias Bound})\label{lem:mdlbias}
	%\begin{equation}
	$\left \lvert \Expc{\bar{H}}-\Ent{p}  \right \rvert \leq \frac{k}{\bSmp}. \nonumber$
	%\end{equation}
\end{lemma}
\begin{proof}
Each iteration of Algorithm~\ref{alg1} chooses $x$ drawn from $p$. Therefore, 
%From Algorithm~\ref{alg1}, we can express $\bar{H}$ as 
% \begin{align*}
% \bar{H} &= \frac{1}{\vSmp}\sum\limits_{t=1}^{\vSmp} \sum\limits_{x\in \mathcal{X}} \indctr{ x = x_t }\log \left( \frac{\bSmp}{N_{x_t} + 1} \right). \\
% &=  \sum\limits_{x\in \mathcal{X}} \frac{1}{\vSmp}\sum\limits_{t=1}^{\vSmp} \indctr{ x = x_t }\log \left( \frac{\bSmp}{N_{x_t} + 1} \right).
% \end{align*}
% The above formulation can be thought of as an empirical average of $\log\left( \frac{\bSmp}{N_X+1} \right)$, where $X \sim p$. 
%Therefore, 
\begin{equation}
\Expc{\bar{H}} =\Expc{\bar{H}_t}=\sum\limits_{x \in [k]} p\left( x \right) \Expc{ \log \left( \frac{\bSmp}{N_{x}+1} \right)},
\end{equation}
where the expectation is over the randomness in $N_x$. 
Therefore, 
\begin{align}
\Ent{p} -  \Expc{ \bar{H}} = \sum\limits_{x \in [k]} p\left(x\right) \log \frac{1}{p\left(x\right)} - \sum\limits_{x\in [k]} p\left(x\right)\Expc{\log \left( \frac{\bSmp}{N_{x}+1} \right)}
= \sum\limits_{x \in [k]} p\left(x \right) \Expc{ \log \left( \frac{N_{x}+1}{\bSmp p\left(x\right)} \right)}.  \nonumber
\end{align}
We now bound this expression. 
\begin{align}
 \sum\limits_{x \in [k]} p\left(x \right) \Expc{ \log \left( \frac{N_{x}+1}{\bSmp p\left(x\right)} \right)}
\stackrel{(a)}{\leq} \sum\limits_{x \in [k]} p\left(x \right)  \log \left( \Expc{\frac{N_{x}+1}{\bSmp p\left(x\right)}}  \right) 
\stackrel{(b)}{=} \sum\limits_{x \in [k]} p\left(x \right)  \log \left( 1 + \frac{1}{\bSmp p\left(x\right)} \right)
\stackrel{(c)}{\leq} \frac{k}{\bSmp},\nonumber
\end{align}
where $(a)$ uses concavity of logarithms, $(b)$ follows since $N_x$ is distributed $Bin(N, p(x))$ and therefore has mean $N p(x)$, and $(c)$ uses $\log (1+x)\le x$. To lower bound the expression, we upper bound its negative. 
%\begin{align}
%\Ent{p} -  \Expc{ \bar{H}} &= \sum\limits_{x \in \mathcal{X}} p\left(x\right) \log \frac{1}{p\left(x\right)} - \sum\limits_{x\in \mathcal{X}} p\left(x\right)\Expc{\log \left( \frac{\bSmp}{N_{x}+1} \right)} \nonumber \\
%&= \sum\limits_{x \in \mathcal{X}} p\left(x \right) \Expc{ \log \left( \frac{N_{x}+1}{\bSmp p\left(x\right)} \right)}  \nonumber \\
%&\leq \sum\limits_{x \in \mathcal{X}} p\left(x \right)  \log \left( \Expc{\frac{N_{x}+1}{\bSmp p\left(x\right)}}  \right) \label{eq:modelJens} \\
%&= \sum\limits_{x \in \mathcal{X}} p\left(x \right)  \log \left( 1 + \frac{1}{\bSmp p\left(x\right)} \right) \nonumber \\
%&\leq \frac{k}{\bSmp}. \label{eq:logineq}
%\end{align}
\begin{align}
\sum\limits_{x \in [k]} p\left(x \right) \Expc{\log \left( \frac{\bSmp p\left(x\right)}{N_{x}+1} \right)}
&\stackrel{(a)}{\leq} \sum\limits_{x \in \mathcal{X}} p\left(x \right)  \log \left( \Expc{\frac{\bSmp p\left(x\right)}{N_{x}+1}} \right) 
\stackrel{(b)}{\leq}  \sum\limits_{x \in \mathcal{X}} p\left(x \right) \log \left( \frac{\bSmp}{\bSmp+1} \right) < 0. \label{eq:invbinom}
\end{align}
where $(a)$ uses concavity of logarithms is obtained using Jensen's inequality and \eqref{eq:invbinom} follows from the following claim plugging in $r=p(x)$, and $m=N$. 

\begin{lemma}\label{binreciplemma}
Let $X \sim \text{Bin} \left(m,r\right)$, then
$
\Expc{\frac{1}{X+1}} \leq \frac{1}{r\left(m+1\right)}.
$%\end{equation*}
\end{lemma}
%\begin{proof}
%\vspace{-\topsep}
\begin{proof}
\[
	\Expc{\frac{1}{X+1}} = \frac1{m+1}\sum\limits_{l=0}^{m} \frac{m+1}{l+1} {{m}\choose{l}} r^l\left(1-r\right)^{m-l}= \frac{1-(1-r)^{m+1}}{r\left(m+1\right)}\leq \frac{1}{r\left(m+1\right)}.\qedhere
\]
\end{proof}
\noindent Combining the upper and lower bound on $\Ent{p} -  \Expc{ \bar{H}}$ proves the lemma.
\end{proof}

%\medskip
\noindent\textit{Concentration.} Using Hoefding's inequality, we prove the following concentration result for $\widehat{H}$.
\begin{lemma}(\textbf{Concentration})
\label{lem:concentration-simple}
For any $\mu>0$,
$ \label{eqn:cnvrge}
\Pr \left( \lvert \bar{H} - \Expc{\bar{H}} \rvert \geq \mu \right) \leq 2\exp \left(- \frac{2R\mu^2}{\log^2 \left(\bSmp+1\right)} \right).$
\end{lemma}
\begin{proof}
In each of the $R$ iterations, $N_x$ takes a value in $\{0, \ldots, N\}$. Therefore, for $t=1, \ldots, R$, $\widehat{H}_t\in[\log (N/(N+1)), \log N]$ are i.i.d. random variables. 
%Note that $\bar{H}$ is an average of $\vSmp$ i.i.d random variables $$Z_t = \sum\limits_{x \in \mathcal{X}} \indctr{x=x_t} \log \left( \frac{\bSmp}{N_{x_t} + 1} \right).  $$ where $t \in \range{\vSmp}$. Each of the $\vSmp$ random variables can take values from $ \range[\log \left( \frac{\bSmp}{\bSmp+1}\right)]{\log \bSmp}$. 
By Hoeffding's inequality (Lemma~\ref{lem:hoeff}), 
\begin{equation} \label{cnvrge}
\Pr \left( \lvert \bar{H} - \Expc{\bar{H}} \rvert \geq \frac{\eps}{2} \right) \leq \exp \left(- \frac{\vSmp \eps^2}{2\log^2 \left(\bSmp+1\right)} \right).\qedhere
\end{equation} 
\end{proof}
%With these two lemmas, we can conclude the sample complexity bounds as follows. 
The choice of $\bSmp$ in Algorithm~\ref{alg1} implies that $\left \lvert \Expc{\bar{H}}-\Ent{p}  \right \rvert \le\eps/2$, and by choosing  $\mu=\eps/2$, and $R =  {4\log^2(1+2k/\eps)}/{\eps^2} $ implies that $\bar H$ is within $H(\p)\pm\eps/2$ with probability at least $2/3$. This gives the total sample complexity of $(\bSmp+1) \vSmp = O\left({k\log^2 \left(k/\eps\right)}/{\eps^3}\right)$.

%% file: two-intervals.tex
\new{The algorithm in the previous section treats each symbol equally and uses the same $N$ and $R$.} 
%To reduce the sample complexity, we make the following thought experiment. Suppose we are guaranteed that all probability values $p(x)$ are at most $(\log k)^\beta/k$ for some positive integer $\beta$. Then, we know that the contribution that each symbol makes to the entropy term is at most $\log k -\beta \log \log k$. We can therefore hope that the number of iterations $R$ needed for these intervals can be made smaller. If there are probability values larger than  $(\log k)^\beta/k$ we should be able to estimate their values with smaller $N$ than needed earlier. 
To reduce the sample complexity, our high level approach is the following:
\begin{itemize}
\item Let $T\in\NN$, and $0=\prt_0<\prt_1<\ldots<\prt_T=1$. We design a partition $\mathcal{I} := \{I_1, I_2, ..., I_T\}$ of $\left[0,1 \right]$ into $T$ intervals with $I_j = [a_{T-j}, a_{T-j+1})$. 
\item
We will express the entropy as an expectation of entropy-like terms over these intervals (Lemma~\ref{lem:decom}), and will estimate the contribution from each interval.
\item Consider the $j$th interval, $I_j = [a_{T-j}, a_{T-j+1})$. For $p(x)\in I_j$, 
%note that $p(x)\in (\log (1/a_{T-j+1}), \log (1/a_{T-j})]$. 
the number of samples needed to estimate $p(x)$ grows roughly as $1/a_{T-j}$. Therefore, intervals close to zero need more samples, and intervals far from zero require fewer samples.  
\item
Note that for $p(x)\in I_j, \log (1/p(x)) \in (\log (1/a_{T-j+1}), \log (1/a_{T-j})]$. We choose the intervals such that this width decreases for intervals close to zero. In doing so, we will ensure that while more samples are needed to estimate the probability values in these intervals, we need fewer iterations (small $R$) to estimate the contribution of these intervals to entropy. 
\end{itemize}

\begin{figure}[H]
\begin{center}
\begin{tikzpicture}[xscale=12]
\draw[-][ultra thick] (0,0) -- (1,0);
\draw [ultra thick] (0,-.1) node[below]{$a_0 = 0$} -- (0,0.1);
\draw [ultra thick] (0.166,-.1) node[below]{$a_1$} -- (0.166,0.1);
\node[below] at (0.333,-.1) {$\cdots$};
%\draw [ultra thick] (0.40,-.1) node[below]{$a_{T-j}$} -- (0.40,0.1);
\draw [ultra thick] (0.50,-.1) node[below]{$a_{T-j}$} -- (0.50,0.1);
\draw [ultra thick] (0.666,-.1) node[below]{$a_{T-j+1}$} -- (0.666,0.1);
,\draw [ultra thick] (0.833,-.1) node[below]{$a_{T-1}$} -- (0.833,0.1);
\draw [ultra thick] (1,-.1) node[below]{$a_T = 1$} -- (1,0.1);
\node[above] at (0.08,0.2) {$I_T$};
\node[above] at (0.57,0.2) {$I_j$};
\node[above] at (0.90,0.2) {$I_1$};
\end{tikzpicture}
\caption{Partition of $\left[0,1\right]$ into $T$ intervals}
\end{center}
\end{figure}
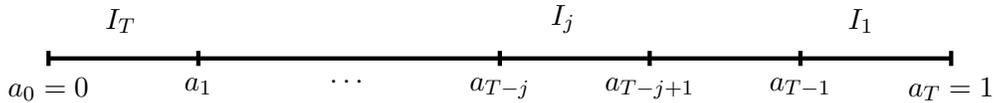
%To reduce the sample complexity, we express $H(p)$ as an expectation of various conditional expectations depending on the symbol probability values. For larger probability values we use smaller $N$ and for small probabilities we use smaller $R$. We then estimate the terms separately to obtain the final estimate.
\noindent In Lemma~\ref{lem:decom} we express entropy as a summation over the contributions from the various intervals. 

\noindent\textbf{Entropy as a Weighted Sum of Conditional Expectations.}\label{subsec:condexp}
Ideally, our approach (similar to the previous section) would be the following. Obtain a symbol in the data stream, find which interval it lies in, and estimate the entropy contribution from each interval separately. However, we may not be able to estimate the exact interval that a symbol is in. To take care of this, consider a randomized algorithm $\mathcal{A}:[k]\to\{I_1,\ldots, I_T\}$ that takes as input $x\in[k]$, and outputs an interval in $\cI$ (which corresponds to our guess for the interval in which $p(x)$ lies in). For a symbol $x$, let $p_{\cA}\left(I_j \middle| x \right) := \probof{\cA(x) = I_j}$ be the distribution of the output of $\mathcal{A}$ for a symbol $x$. For a distribution $\p\in\Delta_k$, let
%\vspace{-5pt}
\begin{align} \label{eq:conddist}
p_\cA\Paren{I_j} := \sum_{x\in [k]}\p(x)\cdot p_{\cA}\left(I_j \middle| x \right), && p_\cA\left(x \middle| I_j \right) := \frac{\p(x)\cdot p_{\cA}\left(I_j \middle| x \right)}{p_\cA\left(I_j \right) }.
\end{align}
Then $p_\cA\Paren{I_j}$ is the probability that $\cA(X)=I_j$, when $X\sim\p$. $p_\cA\left(x \middle| I_j \right)$ is the distribution over $[k]$ conditioned on $\cA(X)= I_j$. . 

For any randomized function $\mathcal{A}:[k]\to\{I_1,\ldots, I_T\}$ we can characterize the entropy as follows. 
\begin{lemma} \label{lem:decom}
	%\begin{align}\vspace{-2ex}
	Let $ H_j := \EE _{X \sim p_\cA \left(x \middle| I_j \right)}\left[ -\log{p\left(X\right)}  \right]$ then, $\Ent{p} = \sum_{j=1}^{T} p_\cA\left(I_j \right)  H_j.$ \label{entintexp}
	%	\end{align}
\end{lemma}
\begin{proof}
\begin{align}
\Ent{\p}  = \sum_x \p(x) \Paren{\sum_j p_{\cA}\left(I_j \middle| x \right)}\log \frac1{\p(x)} &= \sum_x  \sum_j \Paren{p_\cA\left(I_j \right) p_\cA\left(x \middle| I_j \right) \log \frac1{\p(x)}}\label{eqn:mid-step-ent-decomp}\\
&= \sum_j p_\cA\Paren{I_j} \Paren{ \EE _{X \sim p_\cA \left(x \middle| I_j \right)}\left[ -\log{p\left(X\right)}  \right]}.\nonumber
\end{align}	
where~\eqref{eqn:mid-step-ent-decomp} follows from~\eqref{eq:conddist}. 
\end{proof}
Suppose $\cA$ is such that it outputs the exact interval in which $p(x)$ is in, then $p_\cA\Paren{I_j}=p(I_j)$, the total probability of interval $I_j$, and $p_\cA(x|I_j)$ is the conditional distribution of all symbols in $I_j$. In this case, the lemma above reduces to the grouping property of entropy~\cite{CoverT06}. In our streaming setting, the algorithm $\cA$ will take as input an element $x$ of the data stream, and then based on the number of occurrences of $x$ over a window of certain size in the subsequent stream outputs an interval in $\cI$. 

We will choose the intervals and algorithm $\cA$ appropriately. By estimating each term in the summation above, we will design an algorithm with $T$ intervals that uses $O\Paren{\frac{k (\log^{(T)} k + \log(1/\eps))^2}{\eps^3}}$ 
 samples and a constant words of space, and estimates $H(\p)$ to $\pm\eps$ with probability at least $2/3$. Here $\log ^{(T)}$ denotes the iterated logarithms, and therefore shows the improvement in logarithmic terms as $T$ grows. 

In Section~\ref{sec:two_interval}, we provide the details with $T=2$. %, and obtain an algorithm with sample complexity $O(\frac{k (\log (\log (k)/\eps))^2}{\eps^3})$, which improves the sample requirements of Theorem~\ref{thm:one-interval} from a single $\log$, to $\log \log$. 
This section will flesh out the key arguments, and show how to reduce the $\log k$ term in the previous section to $\log \log k$. Finally in Section~\ref{sec:general_interval}, we extend this to $T = \log^* k$ where $\log^* k= \min_i \{\log^{(i)} k \le 1\} $ intervals to further reduce the sample complexity to $O({k (\log (1/\eps))^2}/{\eps^3})$.

\subsection{Two Intervals Algorithm ($T=2$)} \label{sec:two_interval}

We propose Algorithm~\ref{alg:twoint} with the following guarantee. 
%\vspace{-5pt}
\begin{theorem} \label{thm:perf_two}
	Algorithm~\ref{alg:twoint} uses 
	%\begin{align}
	%n_2 &= C_1 \cdot \frac{k}{\eps}, \hspace{5mm} n = n_1 = C_1 \cdot\frac{k}{\eps \left( \log k \right)^{\gamma} }. \label{eqn:defn}\\
	%R_2 &= C_2 \cdot \frac{\left(\log \log k\right)^2}{\epsilon^2}, \hspace{5mm} R = R_1 = C_2 \frac{\log^2 k}{\eps^2}. \label{eqn:defr}
	%\end{align}
	%with proability at least 2/3, we have:
	%\[	|\Ent{p} -  \estII| \le \eps \]
$
		O(NR + N_1R_1 + N_2R_2) = O\left( \frac{k (\log \left(\log (k)/\eps\right))^2}{\eps^3}\right)
$
	samples, $20$ words and outputs an $\pm\eps$ estimate of $H(p)$ with probability at least $2/3$.
\end{theorem}
%\vspace{-10pt}

\subsubsection{Description of the Algorithm}
Let $T=2$, and $\beta>16$ be a constant. Consider the following partition of $\left[0 ,1\right]$:
%\vspace{-10pt}
\begin{equation}
I_2 = \left[0,\ell \right), I_1 = \left[\ell,1\right] \hspace{5mm} \text{where} \hspace{5mm} \ell = {\left(\log k \right)^{\beta}}/{k}.\label{eqn:two-ints}
\end{equation}

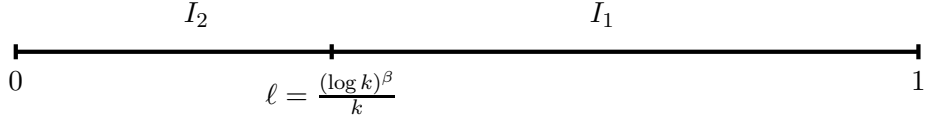
\begin{figure}[H]
\begin{center}
\begin{tikzpicture}[xscale=12]
\draw[-][ultra thick] (0,0) -- (1,0);
\draw [ultra thick] (0,-.1) node[below]{$0$} -- (0,0.1);
\draw [ultra thick] (0.35,-.1) node[below]{$\ell = \frac{(\log k)^{\beta}}{k}$} -- (0.35,0.1);
\draw [ultra thick] (1,-.1) node[below]{$1$} -- (1,0.1);
\node[above] at (0.2,0.2) {$I_2$};
\node[above] at (0.65,0.2) {$I_1$};
\end{tikzpicture}
\caption{Partition of $\left[0,1\right]$ into two intervals}
\end{center}
\end{figure}

We now specify the algorithm $\cA:[k]\to\{I_1, I_2\}$ we used in Lemma~\ref{lem:decom}. $\cA$ is denoted by $\textsc{EstInt}$ (Algorithm~\ref{alg:intvlesti}). %t essentially tries to \emph{guess the interval} in which the probability of $x$ lies. 
For $x\in[k]$, it takes $\bSmp$ samples from $\p$, and if the fraction of occurrences of $x$ is more than $\ell$ it outputs $I_1$, else it outputs $I_2$. This is an algorithm that tries to guess the true interval containing $p(x)$ from the samples.
%\begin{minipage}[t]{0.45\textwidth}
%\null
%\begin{algorithm}[H]
%	\caption{$\cA$ : $\textsc{EstInt} \left(\bSmp, x\right)$} \label{alg:intvlesti}
%	\begin{algorithmic}[1]
%		\Require $\bSmp$, $x\in [k]$
%		\State Obtain $\bSmp$ samples from $p$ 
%		\State { \textbf{if} $x$ appears $\geq \bSmp\ell$ times output $I_1$}
%		\State { \textbf{else} output $I_2$}
%	\end{algorithmic}
%\end{algorithm} 
%\end{minipage}
%%\hfill
%\begin{minipage}[t]{0.45\textwidth}
%\null
%\begin{algorithm}[H]
%	\caption{ $\textsc{EstProbInt}\left(\bSmp, \vSmp\right)$} \label{alg:est_p1}
%	\begin{algorithmic}[1]
%		\Require $\bSmp, \vSmp$
%		\State $\widehat{p}_\cA\left(I_1 \right) = 0$
%		\For{$t = 1 \text{ to } \vSmp$}
%		\State Sample $x \sim \p$.
%		\State {\textbf{if} $\textsc{EstInt}\left(\bSmp, x\right) = I_1$}
%		\State {\textbf{then} $\widehat{p}_\cA\left(I_1 \right)= \widehat{p}_\cA\left(I_1 \right)+ \frac{1}{\vSmp}$}
%		\EndFor
%	\end{algorithmic}
%\end{algorithm} 
%\end{minipage}
\begin{table}[ht]
	\begin{tabular}{cc}
		\begin{minipage}[t]{0.4\textwidth}
		\begin{algorithm}[H]
		\caption{$\cA$ : $\textsc{EstInt} \left(\bSmp, x\right)$} \label{alg:intvlesti}
		\begin{algorithmic}[1]
%		\Require $\bSmp$, $x\in [k]$
		\State Obtain $\bSmp$ samples from $p$ 
		\State { \textbf{if} $x$ appears $\geq \bSmp\ell$ times, output $I_1$}
		\State { \textbf{else} output $I_2$}
		\end{algorithmic}
		\end{algorithm} 
		\end{minipage}
%\hfill
		&
		\begin{minipage}[t]{0.5\textwidth}
%\null
		\begin{algorithm}[H]
		\caption{ $\textsc{EstProbInt}\left(\bSmp, \vSmp\right)$} \label{alg:est_p1}
		\begin{algorithmic}[1]
%		\Require $\bSmp, \vSmp$
		\State $\widehat{p}_\cA\left(I_1 \right) = 0$
		\For{$t = 1 \text{ to } \vSmp$}
		\State Sample $x \sim \p$.
		\If {$\textsc{EstInt}\left(\bSmp, x\right) = I_1$}
		 \State $\widehat{p}_\cA\left(I_1 \right)= \widehat{p}_\cA\left(I_1 \right) + {1}/{\vSmp}$
		 \EndIf
		\EndFor
		\end{algorithmic}
		\end{algorithm} 
		\end{minipage}
\end{tabular}
\end{table}
%We can see if we choose $n$ large enough, then with a high probabiltiy, $EstInt(n,x) = I_1$ if $\p(x) > 2n\ell $ and $EstInt(n,x) = I_2$ if $\p(x) < n\ell/2 $. This seperation will help us reduce the number of samples needed to estimate $H_1$ and $H_2$. Next we will see how to estimate $p_\cA\left(I_1 \right), p_\cA\left(I_2 \right)$  and $H_1, H_2$.

By Lemma~\ref{lem:decom}, $
	\Ent{p} =  p_\cA\left(I_1 \right)  H_1 +p_\cA\left(I_2 \right)  H_2.$
We estimate the terms in this expression as follows. 
%\medskip

\medskip
\noindent\textbf{Estimating $p_\cA(I_j)$'s.} We run $\textsc{EstInt}$ multiple times on samples generated from $\p$, and output the fraction of times the output is $I_j$ as an estimate of $\p_\cA(I_j)$. We only estimate $p_\cA(I_1)$, since $p_\cA(I_1)+p_\cA(I_2)=1$. The complete procedure %for estimating the probabilities are given below, and it takes two parameters $\bSmp$, and $R$. This 
is specified in Algorithm~\ref{alg:est_p1}.
%\begin{algorithm}[ht]
%	\caption{Estimating $p_\cA\left(I_1 \right), p_\cA\left(I_2 \right)$ : $\textsc{EstProbInt}\left(\bSmp, \vSmp\right)$} \label{alg:est_p1}
%	\begin{algorithmic}[1]
%		\Require $\bSmp, \vSmp$
%		\State $\widehat{p}_\cA\left(I_1 \right) = 0$
%		\For{$t = 1 \text{ to } \vSmp$}
%		\State Sample $x \sim \p$.
%		\State {\textbf{if} $\textsc{EstInt}\left(\bSmp, x\right) = I_1$ \textbf{then} $\widehat{p}_\cA\left(I_1 \right)= \widehat{p}_\cA\left(I_1 \right)+ \frac{1}{\vSmp}$}
%		\EndFor
%	\end{algorithmic}
%\end{algorithm} 

\medskip
\noindent\textbf{Estimating $H_j$'s.} Recall that $H_j$'s are the expectations of $-\log \left({p(x)}\right)$ under different distributions given in~\eqref{eq:conddist}. %Even though we do not know these distributions, we can generate samples from them using $\cA$, and then estimate the probability of the sample thus generated by taking more samples from the underlying $\p$.  
Since the expectations are with respect to the conditional distributions, we first sample a symbol from $p$ and then conditioned on the event that $\textsc{EstInt}$ outputs $I_j$, we use an algorithm similar to Algorithm~\ref{alg1} to estimate $\log (1/\p(x))$. The complete procedure is given in Algorithm~\ref{alg:hi_est2}. Notice that when computing $\widehat{H}_2$ in Step~\ref{step:clip}, we clip the $\widehat{H}_2$'s to $\log \frac{1}{4\ell}$ if $N_{x,2} > 4\ell \bSmp_2 -1$. This is done to restrict each $\widehat{H}_2$ to be in the range of $[\log \frac{1}{4\ell}, \log \bSmp_2]$, which helps in obtaining the concentration bounds by bounding the width of the interval for applying Hoeffding's inequality. 
\begin{algorithm}[ht]
	\caption{Estimating $H_1$ and $H_2$ : $\textsc{CondExp}\left(\bSmp_1, \bSmp_2, \vSmp_1, \vSmp_2\right)$} \label{alg:esthi}
	\label{alg:hi_est2}
	\begin{algorithmic}[1]
		\For{$i=1,2$, set $\widehat{H}_i = 0, S_i = 0$, }
		%\State 
		\For{$t=1 \text{ to } R_i$}
		\State Sample $x \sim \p$
		\If{$\textsc{EstInt}\left(\bSmp, x\right)=I_i$} \label{lst:line:esthicond}
		\State $S_i = S_i + 1$
		\State Let $N_{x,i} \leftarrow$ \# occurrences of $x$ in the next $\bSmp_i$ samples  \label{computehl2}
		\State $\widehat{H}_i = \widehat{H}_i + \log \left(\bSmp_i/(N_{x,i} + 1) \right)$ if $i = 1$
		\State $\widehat{H}_i = \widehat{H}_i + \max \left\{\log \left(\bSmp_i/(N_{x,i} + 1)\right), \log \Paren{1/4\ell} \right\} $ if $i = 2$ \label{step:clip}
		\EndIf
		\EndFor
		\State $\bar{H}_i =\widehat{H}_i/S_i$
		\EndFor
	\end{algorithmic}
\end{algorithm} 

\begin{algorithm}[t]
	\caption{Entropy Estimation with constant space: Two Intervals Algorithm} \label{alg:twoint}
	\begin{algorithmic}[1]
	\Require Accuracy parameter $\eps>0, \gamma=\beta/2$, a data stream $X_1, X_2, \ldots\sim p$
		\State Set 
%		\vspace{-10pt}
		\[
		\bSmp = \bSmp_1 = \frac{C_1 k}{\eps \left( \log k \right)^{\gamma} },\ \vSmp = \vSmp_1 = C_2 \frac{\log(k/\eps)^2 }{\eps^2},\ 
%		\]
%		%\vspace{-10pt}
%		\[ 
		\bSmp_2 = C_1 \cdot \frac{k}{\eps},\  \vSmp_2 = C_2 \cdot \frac{\left(\log ((\log k)/\eps)\right)^2}{\eps^2}
		\] \label{step:params}
		%\vspace{-15pt}
		\State $\widehat{p}_\cA\left(I_1 \right)= \textsc{EstProbInt}\left(\bSmp,\vSmp \right)$
		\State $\bar{H}_1, \bar{H}_2 = \textsc{CondExp}\left(\bSmp_1, \bSmp_2, \vSmp_1, \vSmp_2 \right) $ \label{step:est_ent}
		\State $\estII = \widehat{p}_\cA\left(I_1 \right)\bar{H}_1 + (1-\widehat{p}_\cA\left(I_1 \right))\bar{H}_2  $
	\end{algorithmic}
\end{algorithm}

%\medskip
%\noindent\textbf{Memory Requirements for Algorithm~\ref{alg:intvlesti}.}
%One word is required to store the input symbol $x$ and one word each to store the counter and one symbol from the stream of $\bSmp$ samples. Another word is required to count  the number of times $x$ appears in the $\bSmp$ samples. 

\subsubsection{Performance Guarantees}
%\textbf{Memory Requirements.} We can use one word each to store $k$, $\eps$, $\beta$, $C_1$ and $C_2$ and compute $N_1, N_2, R_1, R_2, \ell$ when required.
%%Similar to the one interval case, we use three words each to store $\vSmp_1, \vSmp_2$ and two words each to store $\bSmp_1, \bSmp_2$. We use two words to store $\ell$ up to accuracy $\frac{1}{k}$.
% \textsc{EstInt} uses two words to keep track of the number of occurrences of $x$. For \textsc{EstProbInt}, we use one word to store $x$ and two words to keep track of the final sum $\widehat{p}_{\cA}\left(I_1\right)$ up to accuracy $\eps^2$. We execute \textsc{CondExp} for each interval separately and use one word to store $x$ and two words each to keep track of $S_i$ and $\widehat{H_i}$. We use two words to store the outputs $\bar{H_1}$ and $\bar{H_2}$ and store the final output $\widehat{H}_{II}$ in one of those. Hence, at most 20 words of memory are sufficient.

\textbf{Memory Requirements.} 
Since $R_1,R_2,N_1,N_2$ and $\ell$ are program constants, we compute them on execution
by storing $k$, $\eps, \beta, C$  using four words in total. For simplicity we set $C_1 = C_2 = C$.
\begin{itemize}
\item \textsc{EstInt} uses two words for the counter and two words to keep track of number of appearances of $x$ since $N \le k^2/\eps^2$. These four words are reused on each invocation of  \textsc{EstInt}.
 \item \textsc{EstProbInt} uses three words to store the counter $t$, one word to store $x$ and three words to store the final output since $R \le k^3/\eps^3$.
\item \textsc{CondExp} is executed for each interval separately which allows reusing the memory required for one iteration. We can use the memory reserved for \textsc{EstProbInt} to store the counter $t$ and the sample $x$. $N_{x,i}$'s can be stored in the memory reserved for \textsc{EstInt}. Variables $S_i, \hhat_i$ requires three and two words respectively. The final answer is stored in the memory allocated to $ \hhat_i$.
\end{itemize}
Hence, at most 20 words of memory are sufficient.

\medskip
\noindent \textbf{Sample Complexity.}
%The following theorem establishes sample complexity guarantees on the estimate of Algorithm~\ref{alg:twoint}.
Define Algorithm~\ref{alg:esthi}$^*$ to be a modified version of Algorithm~\ref{alg:esthi} with Step~\ref{step:clip} being 
\[
	\widehat{H}_i = \widehat{H}_i + \log \left(\bSmp_i/(N_{x,i} + 1)\right)
\]
(i.e. without clipping from below) and all other steps remaining the same. Let $\estIIs$ be the output of Algorithm~\ref{alg:twoint} by replacing Step~\ref{step:est_ent} with estimates of Algorithm~\ref{alg:esthi}$^*$. Then we can bound the estimation error by the following three terms and we will bound each of them separately,
\[
	 \left \lvert \Ent{p} -  \estII \right \rvert \le \underbrace{ \left \lvert \Ent{p} -  \Expc{\estIIs} \right \rvert}_\text{Unclipped Bias} + \underbrace{ \left \lvert \Expc{\estII} -  \Expc{\estIIs} \right \rvert}_\text{Clipping Error} + \underbrace{\left \lvert \estII -  \Expc{ \estII}  \right \rvert.
}_\text{Concentration}
\]

\noindent \textbf{Clipping Error.} %We first see how to bound the clipping error term $| \Expc{\estII} -  \Expc{\estIIs} |$ first. 
%By the design of \textsc{CondExp}, we know that $\widehat{H}_2$ is clipped only when the following event occurs for some $x \in \cX$: $ \cE_x= \{ \textsc{EstInt}(\bSmp,x) = I_2, N_{x,2} > 4\bSmp_2 \ell - 1\}$. We can bound the clipping error by showing that $\probof{\cE_x}$ is small. We present the following lemma that bounds the clipping error. (proof in Section~\ref{sec:pfcliptwoint}) 
By the design of \textsc{CondExp}, $\widehat{H}_2$ is clipped only when the event $ \cE_x= \{ \textsc{EstInt}(\bSmp,x) = I_2, N_{x,2} > 4\bSmp_2 \ell - 1\}$ occurs for some $x \in [k]$ . We bound the clipping error in the following lemma (proof in Section~\ref{sec:pfcliptwoint}) by showing that $\probof{\cE_x}$ is small.  
\begin{lemma}(\textbf{Clipping Error Bound})\label{lem:cliptwoint}
	Let  $\estII$ be the entropy estimate of Algorithm~\ref{alg:twoint} and let $\estIIs$ be the entropy estimate of the unclipped version of Algorithm~\ref{alg:twoint}. Then 
	$\left \lvert  \Expc{\estII} -  \Expc{\estIIs} \right \rvert \leq {\eps}/{3}.$
\end{lemma}

%\[
	%\probof{\cE_x} \le \min \{ \probof{EstInt(n,x) = I_2}, \probof{N_{x,2} > 4n_2 \ell - 1}\}
%\]

%By Chernoff bound, if $p(x) > 2 \ell$

%\[ \probof{EstInt(n,x) = I_2} \le \exp\left(-\frac{n \ell}{3}\right)\]

%And if $p(x) < 2 \ell$,

%\[ \probof{N_{x,2} > 4n_2 \ell - 1} \le \exp\left(-\frac{2n_2 \ell}{3}\right)\]

%It is easy to verify that values in~\eqref{eqn:defn} would suffice.
\medskip
\noindent\textbf{Concentration Bound.} To prove the concentration bound, we use Lemma~\ref{lem:concentration} to decompose it into three terms. Each of them can be viewed as the difference between an empirical mean and its true expectation, which can be bounded using concentration inequalities. (proof in Section~\ref{sec:pfconctwoint})

%What remains to be shown is that $|\estII -  \Expc{ \estII} |$ is bounded by $\eps/3$ with probability at least 2/3. First we need the following lemma:

\begin{lemma}(\textbf{Concentration Bound}) \label{lem:concentration}
Let $\estII$ be the entropy estimate of Algorithm~\ref{alg:twoint} and let $\bar{H}_i$ be as defined in Algorithm~\ref{alg:twoint}. Let $p_{\cA}\left(I_i \right)$ be the distribution defined in \eqref{eq:conddist} where $\cA$ is \textsc{EstInt}. Then,
%\vspace{-5pt}
	 \begin{align}
	\left \lvert \Expc{\estII} - \estII \right \rvert \leq \sum\limits_{i=1}^{2}  p_{\cA}\left(I_i\right) \left \lvert \bar{H_i}  - \Expc{\bar{H_i}} \right \rvert + |p_{\cA}\left(I_1\right) - \widehat{p}_{\cA} \left(I_1\right) | |\bar{H}_1 - \bar{H}_2 | \leq {\eps}/{3} \nonumber. 
	\end{align}
\end{lemma}
%Using Lemma~\ref{lem:concentration}, we can decompose the error into three terms and . 

%Moreover, we can bound the concentration error of $\bar{H_1}$ and $\bar{H_2}$ separately by choosing $R_1$ and $R_2$ as we wanted. Next we will try to bound each of the term by $\eps/9$ with high probability.

\newer{We provide a brief outline of the proof below. By the union bound, in order to show that with probability at least $2/3$ the sum is less than $\eps/3$, it is sufficient to show that with probability at most $1/9$, each of the terms is greater than $\eps/9$.

To bound $|p_{\cA}\left(I_1\right) - \widehat{p}_{\cA} \left(I_1\right) | |\bar{H}_1 - \bar{H}_2 |$, we first bound the range of $|\bar{H}_1 - \bar{H}_2 |$ and then use Hoeffding's inequality (Lemma~\ref{lem:hoeff}) to obtain concentration of $\widehat{p}_{\cA} \left(I_1\right)$. To bound  $\left \lvert \bar{H_i}  - \Expc{\bar{H_i}} \right \rvert$, note that we cannot obtain concentration using Hoeffding's inequality because $\vSmp_i$ (the number of samples that we average over) is a random variable. Therefore, we apply Random Hoeffding inequality (Lemma~\ref{lem:ranhoeff}) to $\bar{H_i}$. Since $\vSmp_i$ depends on the range of the random variables being averaged over, we obtain a reduction in the sample complexity for $i=2$ because of clipping the estimate below to $\log \frac{1}{4\ell}$. Therefore, the range for the second interval is $\log(\bSmp_2) - \log \frac{1}{4\ell} = O\left(\log\left( \left(\log k\right)/\eps \right)\right)$ implying $\vSmp_2 = O\left( {(\log \left( (\log  k)/\eps\right))^2}/{\eps^2}\right)$ suffices for the desired probability. For $i=1$, since the range is the same as the one interval case, we use the same $\vSmp_1$ as in the previous section. Note $\vSmp_2 \ll \vSmp_1$.}

\medskip
\noindent\textbf{Bias Bound.} We bound the bias of the unclipped version, $\estIIs$ using the following lemma whose proof is in Section~\ref{sec:pfbiastwo}.

\begin{lemma} (\textbf{Unclipped Bias Bound}) \label{lem:bias_two}
Let $\estIIs$ be the unclipped estimate of Algorithm~\ref{alg:twoint} and let $p_{\cA}\left(I_i \middle| x\right)$ be the conditional distribution defined in \eqref{eq:conddist} where $\cA$ is \textsc{EstProbInt}. Then,
\vspace{-5 pt}
	\begin{align} \label{eq:biastwoint}
	\left| \Ent{p} -  \Expc{\estIIs} \right| \le \sum\limits_{i=1}^{2}  \left( \sum\limits_{x \in \mathcal{X}} { p_{\cA}\left(I_i \middle| x \right)}/{\bSmp_i}\right) \leq {\eps}/{3}.
	\end{align}
\end{lemma}
%\vspace{-5pt}

\newer{Lemma~\ref{lem:bias_two} allows us to choose $\bSmp_1$ and $\bSmp_2$ separately to bound the bias. Interval $I_2$'s contribution is at most $\frac{k}{\bSmp_2}$. For interval $I_1$, we improve upon $\frac{k}{\bSmp_1}$ by partitioning $\cX$ into sets $\cX_1 = \{x \in \cX | p(x) < \ell/2\}$ and $\cX_2  = \{x \in \cX | p(x) \ge \ell/2\}$. For $\cX_1$, $p_{\cA}\left(I_1 \middle| x \right)$ is small by Chernoff bound. For $\cX_2$, since $p(x) \ge \ell/2$, $|\cX_2| \leq 2/\ell$ which is smaller than $k$.  Hence we can choose $N_2 < N_1$.}

In the sample complexity of the two interval algorithm, observe that the term $\bSmp_2\vSmp_2$ dominates. Reducing $\bSmp_2$ is hard because it is independent of the interval length. Therefore we hope to reduce $\vSmp_2$ by partitioning into intervals with smaller lengths. In the smallest interval, if we reduce the range of each estimate to be within a constant, then $O(\frac{1}{\eps^2})$ samples would suffice for concentration. In the next section, we make this concrete by considering an algorithm that uses multiple intervals.

%% file: general-algorithm.tex
\subsection{General Intervals Algorithm}\label{sec:general_interval}

The general algorithm follows the same principles as the previous section with a larger number of intervals, decreasing the sample requirements at each step, as discussed in Section~\ref{sec:techniques}. However, the proofs are much more involved, particularly in order to obtain an $O(k)$ upper bound on the sample complexity. We will sketch some of the key points and move details to the appendix. 

\noindent\textbf{Intervals.} Let $T = \log^* k$, where $ \log^* k:= \min_i \{\log^{(i)} k \le 1\}$. Consider the following partition of $[0, 1]$: $\{I_i\}_{i = 1}^T$ where $I_1 = [l_1, h_1]$ and for $i = 2, ..., T$, $I_i = [l_i, h_i)$,  $h_i = \frac{(\logjabk{i-1})^{\beta}}{\absz} (\beta > 16)$  and $\ell_{i-1} = h_i$. Define $l_T = 0$ and $h_1 = 1$, then we have for $i = 2, ..., T-1$ :
%\vspace{-5pt}
\begin{align}
	I_1 = \left[\frac{(\logjabk{1})^{\beta}}{\absz}, 1\right],  I_\iter = \left[0, \frac{(\logjabk{T-1})^{\beta}}{\absz}\right), I_i = \left[\frac{(\logjabk{i})^{\beta}}{\absz}, \frac{(\logjabk{i-1})^{\beta}}{\absz}\right). \nonumber 
\end{align}

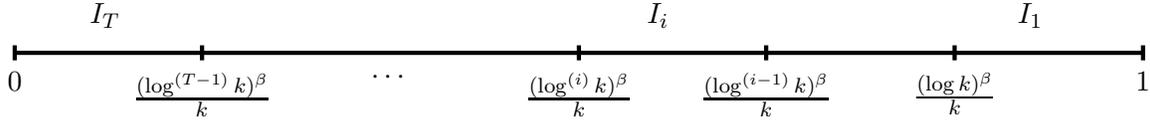
\begin{figure}[H] 
\begin{center}
\begin{tikzpicture}[xscale=15]
\draw[-][ultra thick] (0,0) -- (1,0);
\draw [ultra thick] (0,-.1) node[below]{$0$} -- (0,0.1);
\draw [ultra thick] (0.166,-.1) node[below]{$\frac{(\log^{(T-1)}k)^{\beta}}{k}$} -- (0.166,0.1);
\node[below] at (0.333,-.1) {$\cdots$};
%\draw [ultra thick] (0.40,-.1) node[below]{$a_{T-j}$} -- (0.40,0.1);
\draw [ultra thick] (0.50,-.1) node[below]{$\frac{(\log^{(i)}k)^{\beta}}{k}$} -- (0.50,0.1);
\draw [ultra thick] (0.666,-.1) node[below]{$\frac{(\log^{(i-1)}k)^{\beta}}{k}$} -- (0.666,0.1);
,\draw [ultra thick] (0.833,-.1) node[below]{$\frac{(\log k)^{\beta}}{k}$} -- (0.833,0.1);
\draw [ultra thick] (1,-.1) node[below]{$1$} -- (1,0.1);
\node[above] at (0.08,0.2) {$I_{T}$};
\node[above] at (0.57,0.2) {$I_i$};
\node[above] at (0.90,0.2) {$I_1$};
\end{tikzpicture}
\caption{Partition of $\left[0,1\right]$ into $T = \log^* k$ intervals}
\end{center}
\end{figure}
	Since $T = \log^* k$, we have $I_T \subset \left[0, e^\beta/k\right)$. We divide the bottleneck of the two intervals algorithm $I_2$, into further intervals until the width of the smallest interval is a constant over $k$ ($e^{\beta}/{k}$) which implies concentration with lesser samples than before. Algorithm~\ref{alg:interval-general} defines a distribution over $T$ intervals for each $x$. Using Lemma~\ref{lem:decom}, similar to the two intervals case, we will estimate each of the $p_{\cA} \left(I_i\right)$  and $H_i$'s independently in Algorithm~\ref{alg:estp_general} (\textsc{GenEstProbInt}) and Algorithm~\ref{alg:esth_general} (\textsc{GenCondExp}). Complete algorithm for $T = \log^*k$ is presented in Algorithm~\ref{alg:genint}.
\begin{algorithm}[h]
	\caption{Estimating intervals: General Case : \textsc{GenEstInt}$(\bset{t},x)$} \label{alg:interval-general}
	\begin{algorithmic}[1]
		\Require $\bset{t},x$ drawn from $p$
		\For{$i = 1 \text{ to } t  $}
		\State Generate $\bSmp_i$ samples from $p$ 
		\State {\textbf{if} $x$ appears more than $\bSmp_i \ell_i$ times  \textbf{then} Output $I_i$}
		\EndFor
		%\State Generate $\bSmp_t$ samples from $p$ 
		%\State {\textbf{if} $x$ appears more than $\bSmp_t \ell_t$ times \textbf{then} Output $I_t$}
		%\State{\textbf{else} Output $I_T$}
		\State Output $I_T$
	\end{algorithmic}
\end{algorithm}

\begin{algorithm}[h]
	\caption{Estimating $p_\cA\left(I_i \right),1 \leq i \leq T-1$ : \textsc{GenEstProbInt}$(\bset{T-1},\vset{T-1})$ } \label{alg:estp_general}
	\begin{algorithmic}[1]
		\Require $\bset{T-1},\vset{T-1}$
		\For{$i = 1, 2, ..., T-1$}
		\State $\widehat{p}_\cA\left(I_i\right) = 0$
		\For{$t = 1 \text{ to } \vSmp_i$}
		\State Sample $x \sim \p$.
		\State{\textbf{if} $\textsc{GenEstInt}\left(\left \lbrace \bSmp_j \right \rbrace_{j = 1}^{i} , x\right) = I_i$ \textbf{then} $\widehat{p}_\cA\left(I_i \right)= \widehat{p}_\cA\left(I_i \right)+ \frac{1}{\vSmp_i}$}	
		\EndFor
		\EndFor
		%$\widehat{p}_{\cA} \left(I_T\right) = 1 - \sum\limits_{i=1}^{T-1} 	\widehat{p}_{\cA} \left(I_i\right)$
	\end{algorithmic}
\end{algorithm} 

\begin{algorithm}[h]
	\caption{Estimating ${H}_i$'s : \textsc{GenCondExp}$\left( \bset{T}, \vset{T} \right)$} \label{alg:esth_general}
	\begin{algorithmic}[1]
		\Require $\bset{T},\vset{T}$
		\For{$i=1 \text{ to } T$}
		\State $\widehat{H}_i = 0, S_i = 0$
		\For{$t = 1 \text{ to } \vSmp_i$}
		\State Generate $x \sim p$
		\If{$\textsc{GenEstInt}\left(\left \lbrace \bSmp_j \right \rbrace_{j = 1}^{i}, x\right)$is $I_i$}
		\State $S_i = S_i + 1$
		\State Let $N_{x,i} \leftarrow$ \# occurrences of $x$  in the next $\bSmp_i$ samples
		\State $E_{x,i} = \max \{ \log \left(\frac{\bSmp_i}{N_{x,i} + 1} \right), \log \frac{1}{4h_i}\} $ \label{clipping_gen}
		\State $\widehat{H}_i = \widehat{H}_i +E_{x,i} $ 
		\EndIf
		\EndFor
		\State $\bar{H}_i = \frac{\widehat{H}_i}{S_i}$
		\EndFor
		%	\State Output  $\widehat{H} = \sum_{i = 1}^{T - 1}\widehat{p}\left[I_i\right] \widehat{H_i} + \left(1-\sum_{i = 1}^{T - 1} \widehat{p}\left[I_i\right]\right) \widehat{H_T}$
	\end{algorithmic}
\end{algorithm} 

\begin{algorithm}[ht]
	\caption{Entropy Estimation with constant space: General Intervals Algorithm} \label{alg:genint}
	\begin{algorithmic}[1]
	\Require Accuracy parameter $\eps>0, \gamma=\beta/2$, a data stream $X_1, X_2, \ldots\sim p$.
	
		\State Set 
		\[
		\bSmp_i=C_{\bSmp} \cdot\frac{\absz}{\eps(\logjabk{i})^{\gamma}}, \hspace{5mm} \vSmp_i = C_{\vSmp}\cdot \frac{(\log (\logjabk{i-1}/\eps))^2}{\eps^{2}} \hspace{5mm} 1 \leq i \leq T-1
		\]
		%\vspace{-15pt}
		%\[ 
		%\bSmp_i=C_{\bSmp} \cdot\frac{\absz \log^*\left(\log^*k\right)}{\eps(\logjabk{i})^{\gamma}}, \hspace{5mm} \vSmp_i = C_{\vSmp}\cdot \frac{(\log (\logjabk{i-1}/\eps))^2}{\eps^{2}} \hspace{5mm} T-T' < i \leq T-1
		%\]
		\[ 
		\bSmp_T = C_{\bSmp} \cdot \frac{k}{\eps}, \hspace{5mm} \vSmp_T = C_{\vSmp} \cdot \frac{(\log (\logjabk{T-1}/\eps))^2}{\eps^2}
		\]
		\State $\left \lbrace \widehat{p}_\cA\left(I_i \right)\right \rbrace_{i=1}^{T-1} = \textsc{GenEstProbInt}\left(\bset{T-1},\vset{T-1} \right)$
		\State $\left \lbrace \bar{H}_i \right \rbrace_{i=1}^{T} = \textsc{GenCondExp}\left(\bset{T}, \vset{T}\right) $  \label{step:estcondexp}
		\State $\esti = \sum_{i = 1}^{T-1} \widehat{p}_\cA\left(I_i \right) \bar{H}_i + (1 - \sum_{i = 1}^{T-1} \widehat{p}_\cA\left(I_i \right))\bar{H}_T$
	\end{algorithmic}
\end{algorithm}

%\medskip 
%\noindent\textbf{Memory Requirements of Algorithm~\ref{alg:genint}}
%We reserve one word to compute $\widehat{p}_{\cA} \left(I_T\right)$. Apart from this, we require the same memory as that of the two interval algorithm since \textsc{GenEstInt}, \textsc{GenEstProbInt} and \textsc{GenCondExp} require same memory as that of their two interval counterparts. Also, we perform one iteration each of \textsc{GenEstProbInt} and \textsc{GenCondExp} and maintain running sum of $\widehat{p}_{\cA} \left(I_i\right) \bar{H}_i$'s. Therefore, Algorithm~\ref{alg:genint} requires at most 20 words of space.
%

%\medskip
%\noindent\textbf{Performance Guarantees.}

\medskip
\noindent\textbf{Memory Requirements.}
The analysis of memory requirement is similar to that of the two interval case. To store parameters $\ell_i, \bSmp_i, \vSmp_i$'s, we only store $k, \eps, \gamma, C_N$ and $C_R$ and compute the parameters on the fly. Notice that for each interval, the execution of \textsc{GenEstInt}, \textsc{GenEstProbInt} and \textsc{GenCondExp} require same memory as that of their two interval counterparts. The trick here is that we don't need to store $\widehat{p}_{\cA} \left(I_i\right)$'s and $\bar{H}_i$'s since we can perform each of \textsc{GenEstProbInt} and \textsc{GenCondExp} for one interval and maintain a running sum of $\widehat{p}_{\cA} \left(I_i\right) \bar{H}_i$'s. Therefore, Algorithm~\ref{alg:genint} uses at most 20 words of space.

\medskip
\noindent\textbf{Sample complexity.} Algorithm~\ref{alg:genint} proves the main claim of our paper in Theorem~\ref{thm:main_perf}. \newer{The key idea to remove the extra loglog factor in Theorem~\ref{thm:perf_two} is to progressively make the number of iterations required smaller for the smaller probability intervals.} Similar to the two interval case, we denote Algorithm~\ref{alg:esth_general} without clipping at Step~\ref{clipping_gen}  by Algorithm~\ref{alg:esth_general}$^*$, We further use $\estis$ to represent the final estimate by Algorithm~\ref{alg:genint} with Algorithm~\ref{alg:esth_general} replaced by Algorithm~\ref{alg:esth_general}$^*$ at Step~\ref{step:estcondexp}. Then the error can be bounded by the following three terms:
%\vspace{-5pt}
\begin{align} \label{eqn:gen_decom}
	|\Ent{p} -  \esti| \le \underbrace{| \Ent{p} -  \Expc{\estis} |}_{\text{Unclipped Bias}} + \underbrace{| \Expc{\esti} -  \Expc{\estis} |}_{\text{Clipping Error}} +\underbrace{ |\esti -  \Expc{ \esti} |}_{\text{Concentration}}.
\end{align}

With the parameters defined in Algorithm~\ref{alg:genint}, we can bound the unclipped bias and clipping error in~\eqref{eqn:gen_decom} by $\eps/3$ each and show that the concentration part is also bounded by $\eps/3$ with probability at least $2/3$. The details are given in Lemma~\ref{lem:unclipbias}, \ref{lem:cliperr}, and~\ref{lem:genvar} in Appendix~\ref{pf:thm:general}.

%% file: app-proofs-simple.tex
%\begin{lemma}\label{genbinreciplemma}
%Let $X \sim \text{Bin}\left(n,p\right)$ and $\beta = \frac{1}{\log^2\left(k\right)}$ then 
%\begin{equation}
%\Expc{\frac{1}{X+\beta}} \leq \text{something}
%\end{equation}
%\end{lemma}

\section{Proof of Random Hoeffding Inequality (Lemma~\ref{lem:ranhoeff})}
\label{app:simple}
%\begin{align*}
%\Pr\left(\left \lvert X - \Expc{X} \right \rvert \geq \frac{t}{p} \right) &= \sum\limits_{r=0}^{\infty} \Pr\left(\left \lvert X - \Expc{X} \right \rvert \geq \frac{t}{p} \middle| M=r \right) \Pr\left(M=r\right) \\
%&\leq \sum\limits_{r=0}^{\infty} 2\exp\left(\frac{-2rt^2}{p^2\left(b-a\right)^2} \right) \Pr\left(M=r\right) \\
%&\leq \sum\limits_{r=0}^{\left\lfloor {\frac{mp}{2}} \right \rfloor} 2\exp\left(\frac{-2rt^2}{p^2\left(b-a\right)^2} \right) \Pr\left(M=r\right) + \sum\limits_{r=\left \lceil {\frac{mp}{2}} \right \rceil}^{\infty} 2\exp\left(\frac{-2rt^2}{p^2\left(b-a\right)^2} \right) \Pr\left(M=r\right)\\
%&\leq \Pr\left(M\leq \frac{mp}{2} \right) + 2\exp\left( \frac{-mt^2}{p\left(b-a\right)^2}\right) \\
%&\leq \exp\left(\frac{-mp}{8} \right) + 2\exp\left( \frac{-mt^2}{p\left(b-a\right)^2}\right).
%\end{align*}
Since $X_i\in[a,b]$, we have $\left \lvert X - \Expc{X} \right \rvert \le b - a$. If $t > p (b - a)$, the left hand side is zero and the inequality  holds. We assume $t \le  p (b - a)$, which is equivalent to $p \ge \frac{t^2}{p(b-a)^2}$. 
\begin{align*}
\Pr\left(\left \lvert X - \Expc{X} \right \rvert \geq \frac{t}{p} \right) &= \sum\limits_{r=0}^{m} \Pr\left(\left \lvert X - \Expc{X} \right \rvert \geq \frac{t}{p} \middle| M=r \right) \Pr\left(M=r\right).
\end{align*}
We divide the summation into two parts, $r \le \left \lfloor {\frac{mp}{2}} \right \rfloor $ and $r \ge \left \lceil {\frac{mp}{2}} \right \rceil$. For the first part, by Chernoff bound,
\begin{align*}
	\sum\limits_{r=0}^{\left\lfloor {\frac{mp}{2}} \right \rfloor} \Pr\left(\left \lvert X - \Expc{X} \right \rvert \geq \frac{t}{p} \middle| M=r \right) \Pr\left(M=r\right) \le \probof{M \le \frac{mp}{2}} \le \exp\left(\frac{-mp}{8} \right).
\end{align*}
For the second part, by Hoeffding Inequality (Lemma~\ref{lem:hoeff}), \begin{align*}
	 \sum\limits_{r=\left \lceil {\frac{mp}{2}} \right \rceil}^{m} 2\exp\left(\frac{-2rt^2}{p^2\left(b-a\right)^2} \right) \Pr\left(M=r\right) \le 2\exp\left( \frac{-2t^2}{p^2\left(b-a\right)^2} \frac{mp}{2}\right) \le 2\exp\left( \frac{-mt^2}{p\left(b-a\right)^2}\right).
\end{align*}
Combining the two, we get
\begin{align*}
	\Pr\left(\left \lvert X - \Expc{X} \right \rvert \geq \frac{t}{p} \right) \le  \exp\left(\frac{-mp}{8} \right) + 2\exp\left( \frac{-mt^2}{p\left(b-a\right)^2}\right) \le 3\exp\left( \frac{-mt^2}{8p\left(b-a\right)^2}\right),
\end{align*}
where the last part uses the bound on $t$.

%% file: two-intervalspf.tex
\subsection{Expectation of Unclipped Version Estimates}\label{sec:expcondent}
Let $S_i$ be the number of times \textsc{EstInt} $= I_i$ during the $R_i$ iterations for interval $I_i$ in \textsc{CondExp}. Let $S_{x,i}$ be the number of times symbol $x$ is the first sampled element among these. Note that $S_i \sim \text{Bin} \left(\vSmp_i,p_{\cA}(I_i)\right)$  and $S_{x,i} \sim \text{Bin}\left(S_i,p_{\cA} \left(x \mid I_i \right)\right)$ where $p_{\cA} \left(x \mid I_i \right) =  \dfrac{p\left(x\right)p_{\cA} \left(I_i \mid x\right)}{p_{\cA} \left(I_i\right)}$. Let $N_{x,i,v}$ be $N_{x,i}$ (defined in \textsc{CondExp}) when $x$ is sampled and \textsc{EstInt}$(N,x) = I_i$ for the $v^{\text{th}}$ time. 
Denote the unclipped version of $\bar{H}_i$ by $\bar{H}^*_i$. We can write $\bar{H}^*_i$ as follows
\begin{equation}\label{eq:condentropy}
\bar{H}^*_i = \frac{1}{S_i}\sum\limits_{x \in \mathcal{X}} \sum\limits_{v=1}^{S_{x,i}} \log \left(\frac{\bSmp_i}{N_{x,i,v}+1} \right).
\end{equation}
The above equation implies that $\bar{H}^*_i$ is an empirical mean of  $\log \left(\frac{\bSmp_i}{N_{X,i}+1} \right)$ where $X \sim p_{\cA} \left(x \mid I_i \right)$. Note that for a fixed $x$, $S_{x,i} \sim \text{Bin}\left(S_i, p_{\cA}(x|I_i) \right)$. Therefore, the expectation is 

 %Note that  $N_{x,i}$ is the number of times $x$ appears in $\bSmp$ samples from $p$, therefore for a fixed $x$, $N_{x,i} \sim \text{Bin}\left(\bSmp_i, p_{\cA} \right)$. Therefore,   where the underlying distribution is $p_{\cA} \left(x \mid I_i \right)$ and its expectation is given by

\begin{equation}\label{eq:condent}
\Expc{\bar{H}_i^*} = \sum\limits_{x \in \mathcal{X}} \frac{p\left(x\right) p_{\cA}\left(I_i \middle| x \right)}{p_{\cA}\left(I_i\right)} \Expc{\log \left( \frac{\bSmp_i}{N_{x,i}+1} \right)}. 
\end{equation}
\subsection{Proof of Lemma~\ref{lem:bias_two} : Unclipped Bias Bound} \label{sec:pfbiastwo}
Define $\bar{H}^*_1$ and $\bar{H}^*_2$ to be the analog of $\bar{H}_1$ and $\bar{H}_2$ in the unclipped version of Algorithm~\ref{alg:twoint}. We first note that 
\begin{equation}\label{expcest} 
\Expc{\estIIs} = p_{\cA}(I_1)\Expc{\bar{H}^*_1} + \left(1-p_{\cA}(I_1)\right)\Expc{\bar{H}^*_2}.
\end{equation}
The above is true since, $\Expc{\widehat{p}_{\cA}(I_1)} =  p_{\cA}(I_1)$ and Algorithm~\ref{alg:esthi} estimates $\widehat{p}_{\cA}\left(I_1\right)$ and $\bar{H}^*_1,\bar{H}^*_2$ independently. \par 

%Clearly,  $N_{x,i,v}$ is independent of $S_{x,i}, S_i$. 
We use the following result from equation~\ref{eq:condent} in Section~\ref{sec:expcondent}
\begin{equation} 
\Expc{\bar{H}^*_i}  = \sum\limits_{x \in [k]} \frac{p \left(x\right) p_{\cA} \left(I_i \middle| x \right)}{p_{\cA} \left( I_i\right) } \Expc{\log \left( \frac{\bSmp_i}{N_{x,i}+1} \right)} .
\end{equation}
 
 Using Lemma \eqref{entintexp} and Jensen's inequality, we have
 \begin{align}
 \Expc{\estIIs} - \Ent{p} &\leq \sum\limits_{i=1}^{2} p_{\cA} \left( I_i\right) \left( \sum\limits_{x \in \mathcal{X}} \frac{p\left(x \right) p_{\cA} \left(I_i \middle| x \right)}{p_{\cA} \left( I_i\right)}\Expc{  \log \left( \frac{\bSmp_i p\left(x\right)}{N_{x,i}+1} \right)} \right) \nonumber  \\
 &\leq \sum\limits_{i=1}^{2} p_{\cA} \left( I_i\right) \left( \sum\limits_{x \in \mathcal{X}} \frac{p\left(x \right) p_{\cA} \left(I_i \middle| x \right)}{p_{\cA} \left( I_i\right)} \log \left( \Expc{ \frac{\bSmp_i p\left(x\right)}{N_{x,i}+1} }\right) \right) \nonumber  \\
 &\leq \sum\limits_{i=1}^{2} p_{\cA} \left( I_i\right) \left( \sum\limits_{x \in \mathcal{X}} \frac{p\left(x \right) p_{\cA} \left(I_i \middle| x \right)}{p_{\cA} \left( I_i\right)} \log \left( \frac{\bSmp_i }{\bSmp_i+1} \right) \right) \leq 0. \label{eq:negbias}
 \end{align}
 where \eqref{eq:negbias} follows from Lemma~\ref{binreciplemma}. To bound the reverse, using~\eqref{entintexp}, Jensen's inequality and the fact that $\log (1+x) \leq x $, we have
 \begin{align}
 \Ent{p} - \Expc{\estIIs} &= \sum\limits_{i=1}^{2} p_{\cA} \left( I_i\right) \left( \sum\limits_{x \in \mathcal{X}} \frac{p\left(x \right)p_{\cA} \left(I_i \middle| x \right)}{p_{\cA} \left( I_i\right)}\Expc{  \log \left( \frac{N_{x,i}+1}{\bSmp_i p\left(x\right)} \right)} \right) \nonumber \\
&\leq \sum\limits_{i=1}^{2} p_{\cA} \left( I_i\right) \left( \sum\limits_{x \in \mathcal{X}} \frac{p\left(x \right) p_{\cA} \left(I_i \middle| x \right)}{p_{\cA} \left( I_i\right)} \log \left( \Expc{  \frac{N_{x,i}+1}{\bSmp_i p\left(x\right)} } \right) \right)  \nonumber \\
&= \sum\limits_{i=1}^{2} p_{\cA} \left( I_i\right) \left( \sum\limits_{x \in \mathcal{X}} \frac{p\left(x \right)p_{\cA} \left(I_i \middle| x \right)}{p_{\cA} \left( I_i\right)} \log \left(  \frac{\bSmp_i p\left(x\right)+1}{\bSmp_i p\left(x\right)} \right) \right)  \nonumber \\
&\leq \sum\limits_{i=1}^{2} p_{\cA} \left( I_i\right)  \left( \sum\limits_{x \in \mathcal{X}} \frac{ p_{\cA} \left(I_i \middle| x \right)}{\bSmp_i p_{\cA} \left( I_i\right)}\right)  \nonumber \\
&= \sum\limits_{i=1}^{2}  \left( \sum\limits_{x \in \mathcal{X}} \frac{ p_{\cA} \left(I_i \middle| x \right)}{\bSmp_i}\right).  \label{biasbd}
 \end{align}

For interval $I_1$, we partition $\cX$ into two sets $\cX_1 = \{x \in \cX | p(x) < \ell/2\}$ and $\cX_2  = \{x \in \cX | p(x) \ge \ell/2\}$. For $ x \in \cX_1$, the probability that algorithm $\textsc{EstInt}(\bSmp,x) = I_1$ is small. In particular, by Chernoff bound, 

\begin{equation}
p_{\cA}\left(I_1 \middle| x \right)= \Pr\left(N_x > \bSmp_1 \ell \right) \leq \exp\left(-\frac{\bSmp_1 \ell}{6}\right).
\end{equation}

For $ x \in \cX_2$, since $p\left(x\right) \ge \ell/2$, $| \cX_2 | \leq \frac{2}{\ell}$ and each $p_{\cA}\left(I_1 \middle| x \right) \le 1$, we have
 \begin{align}
\sum\limits_{x \in \mathcal{X}} \frac{p_{\cA}\left(I_1 \middle| x \right)}{\bSmp_1} &= \sum\limits_{x \in \cX_1} \frac{p_{\cA}\left(I_1 \middle| x \right)}{\bSmp_1} + \sum\limits_{x \in \cX_2} \frac{p_{\cA}\left(I_1 \middle| x \right) }{\bSmp_1}  \le \frac{k}{\bSmp_1}  \exp\left(-\frac{\bSmp_1 \ell}{6} \right)  + \frac{2}{\bSmp_1 \ell}. \label{n1scale} 
\end{align}  
For interval $I_2$, we simply bound each term by 1 and get
\begin{align}
	\sum\limits_{x \in \mathcal{X}} \frac{p_{\cA}\left(I_2 \middle| x \right)}{\bSmp_2}  \le \frac{k}{\bSmp_2}. \label{eqn:n2scale}
\end{align}

%Recall that $\ell = \frac{\left(\log k \right)^{\beta}}{k}$. 
%We verify that values in Algorithm~\ref{alg:twoint} will give us a clipping error bound of $\eps/3$ in Section~\ref{sec:pfbiastwo}.
Plugging in the values of $\bSmp_1, \bSmp_2$ defined in Algorithm~\ref{alg:twoint}, it is easy to see there exists a constant $C_1$ such that~\eqref{eqn:n2scale} and~\eqref{n1scale} are bounded above by $\frac{\eps}{6}$ which completes the proof.

\subsection{Proof of Lemma~\ref{lem:cliptwoint} : Clipping Error Bound} \label{sec:pfcliptwoint}
Define $\bar{H}^*_1$ and $\bar{H}^*_2$ to be the analogue of $\bar{H}_1$ and $\bar{H}_2$ in the unclipped version of Algorithm~\ref{alg:esthi}. Using \eqref{expcest} and the fact that the clipping step is applied only when computing $\bar{H}_2$, we have 

%\begin{equation}
%\left \lvert  \Expc{\estII} -  \Expc{\estIIs} \right \rvert  \leq p_{\cA} \left( I_2\right) \left \lvert \Expc{\bar{H}_2} - \Expc{\bar{H}^*_2} \right \rvert \le \left \lvert \Expc{\bar{H}_2} - \Expc{\bar{H}^*_2} \right \rvert.
\begin{equation}
\left \lvert  \Expc{\estII} -  \Expc{\estIIs} \right \rvert  \leq p_{\cA} \left( I_2\right) \left \lvert \Expc{\bar{H}_2} - \Expc{\bar{H}^*_2} \right \rvert.
\end{equation}
 From Algorithm~\ref{alg:twoint}, we note that $\bar{H}_2$ is different from $\bar{H}^*_2$ only when $ \cE_x= \{ \textsc{EstInt}(\bSmp,x) = I_2, N_{x,2} > 4\bSmp_2 \ell - 1\}$ occurs. Therefore from \eqref{eq:condent}, we have the following
\begin{align}
 \left \lvert \Expc{\bar{H}_2} - \Expc{\bar{H}^*_2} \right \rvert &\leq \sum\limits_{x \in \cX } \probof{N_{x,2} > 4\bSmp_2 \ell - 1} \frac{p\left(x\right) p_{\cA}\left(I_2 \middle| x \right)}{p_{\cA}\left(I_2 \right)} \left( \log \left( \frac{1}{4 \ell} \right) -  \log \left( \frac{\bSmp_2}{\bSmp_2+1}\right) \right) \nonumber \\
&\leq \sum\limits_{x \in \cX}  \probof{N_{x,2} > 4\bSmp_2 \ell - 1} \frac{p\left(x\right) p_{\cA}\left(I_2 \middle| x \right)}{p_{\cA}\left(I_2\right)} \log k. 
\end{align}
If $p(x) > 2 \ell$, by Chernoff bound,

\[ \probof{\textsc{EstInt}(\bSmp,x) = I_2} \le \exp\left(-\frac{\bSmp \ell}{3}\right).\]
If $p(x) < 2 \ell$,
\[ \probof{N_{x,2} > 4\bSmp_2 \ell - 1} = p_{\cA}\left(I_2 \middle| x \right) \le \exp\left(-\frac{2\bSmp_2 \ell}{3}\right).\]

Therefore, plugging in values of $N$ and $N_2$, we have
\begin{equation}
p_{\cA}\left( I_2\right) \left \lvert \Expc{\bar{H}_2} - \Expc{\bar{H}^*_2} \right \rvert \leq \min \left \lbrace \exp\left(-\frac{\bSmp \ell}{3}\right), \exp\left(-\frac{2\bSmp_2 \ell}{3}\right) \right \rbrace \log k \leq \frac{\eps}{3}.
\end{equation}
%\[
	%\probof{\cE_x} \le \min \{ \probof{EstInt(n,x) = I_2}, \probof{N_{x,2} > 4n_2 \ell - 1}\}
%\]

%By Chernoff bound, if $p(x) > 2 \ell$

%\[ \probof{EstInt(n,x) = I_2} \le \exp\left(-\frac{n \ell}{3}\right)\]

%And if $p(x) < 2 \ell$,

%\[ \probof{N_{x,2} > 4n_2 \ell - 1} \le \exp\left(-\frac{2n_2 \ell}{3}\right)\]

%It is easy to verify that values in~\eqref{eqn:defn} would suffice.

 \subsection{Proof of Lemma~\ref{lem:concentration} : Concentration Bound}\label{sec:pfconctwoint}

   Using \eqref{expcest}, we have
 \begin{align}
 \left \lvert \Expc{\estII} - \estII \right \rvert &= \left \lvert p_{\cA} \left(I_1\right) \Expc{\bar{H_1}} + p_{\cA}\left(I_2\right) \Expc{\bar{H_2}} -  \widehat{p}_{\cA} \left(I_1\right) \bar{H}_1 - \widehat{p}_{\cA} \left(I_2\right) \bar{H_2} \right \rvert \nonumber \\
  \begin{split}
 &=\left \lvert p_{\cA} \left(I_1\right) \Expc{\bar{H_1}} + p_{\cA} \left(I_2\right) \Expc{\bar{H_2}} - p_{\cA} \left(I_1\right) \bar{H}_1 -p_{\cA} \left(I_2\right) \bar{H}_2  \right. \nonumber \\
 &\left. \qquad +p_{\cA} \left(I_1\right) \bar{H}_1 + p_{\cA} \left(I_2\right) \bar{H}_2 - \widehat{p}_{\cA} \left(I_1\right) \bar{H_1} -\widehat{p}_{\cA}\left(I_2\right) \bar{H_2} \right \rvert 
 \end{split} \nonumber \\
   \begin{split}
 &=\left \lvert p_{\cA} \left(I_1\right) (\Expc{\bar{H_1}}-\bar{H}_1) + p_{\cA} \left(I_2\right) (\Expc{\bar{H_2}} -\bar{H}_2)  \right. \nonumber \\
 &\left. \qquad +(p_{\cA} \left(I_1\right) - \widehat{p}_{\cA} \left(I_1\right)) \bar{H_1}+ (p_{\cA} \left(I_2\right)  -\widehat{p}_{\cA}\left(I_2\right)) \bar{H_2} \right \rvert 
 \end{split} \nonumber \\
 \begin{split}
 &\le\left \lvert p_{\cA} \left(I_1\right) (\Expc{\bar{H_1}}-\bar{H}_1)\rvert + \lvert p_{\cA} \left(I_2\right) (\Expc{\bar{H_2}} -\bar{H}_2) \rvert \right. \nonumber \\
 &\left. \qquad +\lvert(p_{\cA} \left(I_1\right) - \widehat{p}_{\cA} \left(I_1\right)) \bar{H_1}+ (p_{\cA} \left(I_2\right)  -\widehat{p}_{\cA}\left(I_2\right)) \bar{H_2} \right \rvert 
 \end{split}\\
 &= \sum\limits_{i=1}^{2}  \left \lvert p_{\cA} \left(I_i\right) \left( \Expc{\bar{H_i}} - \bar{H_i} \right) \right \rvert + \left \lvert \left(p_{\cA} \left(I_1\right) - \widehat{p}_{\cA} \left(I_1\right) \right) \left( \bar{H}_1 - \bar{H}_2 \right) \right \rvert\label{threeterms},
\end{align}
where the inequality is from the triangle inequality, and~\eqref{threeterms} is true because $p_{\cA}\left(I_1\right)+p_{\cA}\left(I_2\right) = \widehat{p}_{\cA} \left(I_1\right)+\widehat{p}_{\cA} \left(I_2\right)=1$, implying that
$	p_{\cA}\left(I_1\right)  - \widehat{p}_{\cA} \left(I_1\right)=  - (p_{\cA} \left(I_2\right) - \widehat{p}_{\cA}\left(I_2\right))$.

We first bound $|p_{\cA}\left(I_1\right) - \widehat{p}_{\cA} \left(I_1\right) | |\bar{H}_1 - \bar{H}_2 |$. Note that $\bar{H}_1 \in \left[ \log \frac{\bSmp_1}{\bSmp_1 + 1} ,\log \bSmp_1 \right]$. Due to clipping, $\bar{H}_2 \in \left[ \log \frac{\bSmp_2}{4\bSmp_2 \ell + 1}, \log \bSmp_2 \right]$. Since $\bSmp_2 > \bSmp_1$, $\left \lvert \widehat{H}_1 - \widehat{H}_2 \right \rvert \leq \log \frac{\bSmp_2 \left(\bSmp_1+1\right)}{\bSmp_1}$. Since $\widehat{p}_{\cA} \left(I_1\right)$ is the average of $R$ i.i.d binary random variables with mean $p_{\cA}\left( I_1 \right)$, by Hoeffding's inequality (Lemma~\ref{lem:hoeff}), we have 
\begin{equation}\label{condRhoeff}
\probof{ | p_{\cA}\left( I_1 \right) - \widehat{p}_{\cA}\left(I_1\right) | > t }\leq 2\exp\left( -2\vSmp t^2\right). \nonumber
\end{equation}

Let $t = \frac{\eps}{ 9 \log \frac{\bSmp_2 \left(\bSmp_1+1\right)}{\bSmp_1} }$. There exists constant $C_2$ such that for the value of $\vSmp_1$ from Algorithm~\ref{alg:twoint}, with probability at least $8/9$,
\[
	|{p}_{\cA} \left(I_1\right)  - \widehat{p}_{\cA} \left(I_1\right) | |\bar{H}_1 - \bar{H}_2 |  \le \eps/9.
\]

We cannot directly use Hoeffding's inequality bound $\left \lvert \bar{H_i}  - \Expc{\bar{H_i}} \right \rvert$  since the number of samples that we are taking an average over is a random variable. We therefore apply the Random Hoeffding inequality (Lemma~\ref{lem:ranhoeff}) to $\bar{H_1}$ to get
%\vspace{-10pt}
%\begin{equation}\label{eq:twointRH}
%	\probof{ {p}_{\cA} \left(I_i\right) |\bar{H}_i - \Expc{\bar{H}_i}|  > \eps/9} \leq 2\exp\left(\frac{-R_i \eps^2}{9 p_{\cA}\left(I_i\right) \left(b_i-a_i\right)^2} \right) + \exp\left(-\frac{R_i p_{\cA}\left(I_i\right)  }{8} \right),
%\end{equation}
\begin{equation}\label{eq:twointRH}
\probof{ {p}_{\cA} \left(I_i\right) |\bar{H}_i - \Expc{\bar{H}_i}|  > \eps/9} \leq 3\exp\left(\frac{-R_i \eps^2}{72 p_{\cA}\left(I_i\right) \left(b_i-a_i\right)^2} \right),
\end{equation}
where $[a_i, b_i]$ is the possible range of each independent variables when estimating $\bar{H}_i$. Since $a_1 = \log \left(\frac{\bSmp_1}{\bSmp_1 + 1} \right), b_1 = \log \left(\bSmp_1 \right)$, $b_1 - a_1 = \log (\bSmp_1 + 1) = O(\log \frac{k}{\eps})$. %\bcolor{Since $|\bar{H}_i - \Expc{\bar{H}_i}| \leq b_1 - a_1$, we can assume that $p_{\cA} \left(I_i\right) \geq \frac{\eps/9}{b_1 - a_1}$; since otherwise $\probof{ {p}_{\cA} \left(I_i\right) |\bar{H}_i - \Expc{\bar{H}_i}|  > \eps/9} = 0$. This implies that upto constant factors, it is sufficient to upper bound the first term on the right hand side of \eqref{eq:twointRH}.} 
Therefore, there exists a constant $C_2$ such that $R_1 = C_2 \frac{\log^2 (k/ \eps)}{\eps^2}$ suffices for success probability to be at least $8/9$.

The reduction in sample complexity is obtained for $i = 2$. Here $a_2 = \log \frac{1}{4\ell} $ instead of  $ \log \left(\frac{\bSmp_2}{\bSmp_2 + 1} \right)$ because of the clipping step. Since $b_2 = \log \left(\bSmp_2 \right)$, $b_2 - a_2 = \log (4 \bSmp_2 \ell)= O\left(\log\left( \left(\log k\right)/\eps \right)\right)$. Therefore, $\exists$ constant $C_2$, such that $R_2 = C_2 \cdot \frac{\left(\log \left( (\log k)/\eps\right) \right)^2}{\eps^2} $ would suffice to get a probability at least $8/9$.

%% file: general-bias.tex
We now bound the bias of the unclipped version of the entropy estimate. 
\begin{lemma} (\textbf{Unclipped Bias bound})  \label{lem:unclipbias}
	Let $\estis$ be the entropy estimate given by Algorithm~\ref{alg:genint} without the clipping step in Algorithm~\ref{alg:esth_general} , then 
	\begin{equation}
	\left \lvert \Expc{\estis} - \Ent{p} \right \rvert \leq \frac{\eps}{3}.
	\end{equation}
\end{lemma}

\begin{proof}
%  Denote the unclipped versions of $\bar{H}_i$  by $\bar{H}_i^*$. For interval $I_i$ and $x \in \mathcal{X}$, let $S_{x,i}$ be the number of times symbol $x$ appeared among the $S_i$ runs of \textsc{GenCondExp}. Note that $S_i \sim \text{Bin} \left(\vSmp_i,p_{\cA}\left(I_i\right)\right)$  and $S_{x,i} \sim \text{Bin}\left(S_i,p_{\cA}\left(x \mid I_i \right)\right)$. Let $N_{x,i,v}$ be $N_{x,i}$ (defined in \textsc{GenCondExp}) when $x$ appeared for the $v^{\text{th}}$ time in the $S_{x,i}$ runs of \textsc{GenCondExp} mentioned above. We can write $\bar{H}_i^*$ as follows.
  Denote the unclipped versions of $\bar{H}_i$  by $\bar{H}_i^*$. For interval $I_i$, let $S_i$ be the number of times $\textsc{GenEstInt}\left(\left \lbrace \bSmp_j \right \rbrace_{j = 1}^{i}, x\right) = I_i$ during $R_i$ iterations in Algorithm~\ref{alg:esth_general}. For $x \in \mathcal{X}$, let $S_{x,i}$ be the number of times symbol $x$ is the first sampled element among these. Note that $S_i \sim \text{Bin} \left(\vSmp_i,p_{\cA}\left(I_i\right)\right)$  and $S_{x,i} \sim \text{Bin}\left(S_i,p_{\cA}\left(x \mid I_i \right)\right)$. Let $N_{x,i,v}$ be $N_{x,i}$ (defined in \textsc{GenCondExp}) when $x$ is first sampled and $\textsc{GenEstInt}\left(\left \lbrace \bSmp_j \right \rbrace_{j = 1}^{i}, x\right) = I_i$ for the $v^{\text{th}}$ time. We can write $\bar{H}_i^*$ as follows.
\begin{equation}
\bar{H}_i^* = \frac{1}{S_i}\sum\limits_{x \in \mathcal{X}} \sum\limits_{v=1}^{S_{x,i}} \log \left(\frac{\bSmp_i}{N_{x,i,v}+1} \right).
\end{equation} 
 
Since  $\widehat{p}_{\cA}\left( I_i \right)$ and $\bar{H}_i^*$ are independent, we have 
\begin{equation}
\Expc{\estis}= \sum\limits_{i=1}^{T} \Expc{\widehat{p}_{\cA}\left( I_i\right)} \Expc{\bar{H}^*_i} = \sum\limits_{i=1}^{T} p_{\cA}\left(I_i\right) \Expc{\bar{H}^*_i}.
\end{equation}
For the interval $I_i$, $\Expc{\bar{H}_i^*}$ can be written as (refer Section~\ref{sec:expcondent} for detailed argument):
\begin{align*}
\Expc{\bar{H}_i^*} = \sum\limits_{x \in \mathcal{X}} \frac{p\left(x\right) p_{\cA}\left(I_i \middle| x \right)}{p_{\cA}\left(I_i\right)} \Expc{\log \left( \frac{\bSmp_i}{N_{x,i}+1} \right)}.
\end{align*}
Similar to Equations \eqref{eq:negbias} and \eqref{biasbd}, we have
\begin{align*}
\Expc{\estis} - \Ent{p} &\leq \sum\limits_{i=1}^{T}  p_{\cA}\left(I_i \right) \left( \sum\limits_{x \in \mathcal{X}} \frac{p\left(x \right) p_{\cA} \left(I_i \middle| x \right)}{p_{\cA} \left(I_i\right)} \log \left( \frac{\bSmp_i }{\bSmp_i+1} \right) \right) \leq 0, \\
\Ent{p} - \Expc{\estis} &\leq \sum\limits_{i=1}^{T}  \left( \sum\limits_{x \in \mathcal{X}} \frac{ p_{\cA}\left(I_i \middle| x \right)}{\bSmp_i}\right). 
\end{align*}
Therefore,
\begin{equation}\label{eq:genbiasfinal}
\left \lvert \Ent{p} - \Expc{\estis} \right \rvert \leq \sum\limits_{i=1}^{T} \left( \sum\limits_{x\in \mathcal{X}} \frac{p_{\cA} \left(I_i \middle| x \right)}{\bSmp_i} \right). 
\end{equation}
For interval $I_i$, $1 \leq i \leq T-1$, we divide the symbols into $\mathcal{X}_l = \left \lbrace x : p_x \leq \frac{l_i}{2} \right \rbrace$ and $\mathcal{X}_m = \left \lbrace x : p_x > \frac{l_i}{2} \right \rbrace$ and get

\begin{align}
\sum\limits_{x \in \mathcal{X}} \frac{p_{\cA} \left(I_i \middle| x \right)}{\bSmp_i} &\leq \sum\limits_{x \in \cX_l} \frac{p_{\cA} \left(I_i \middle| x \right)}{\bSmp_i} + \sum\limits_{x \in \cX_m} \frac{p_{\cA} \left(I_i \middle| x \right)}{\bSmp_i} \nonumber \\
&\leq \frac{1}{\bSmp_i} \left( \sum\limits_{x \in \cX_l} \exp\left(-\frac{\bSmp_i l_i}{6}\right) \right) + \frac{2}{\bSmp_i l_i} \nonumber \\
&\leq \frac{k}{\bSmp_i}\exp\left(-\frac{\bSmp_i l_i}{6}\right) + \frac{2}{\bSmp_i l_i}. \label{eq:gennicond}
\end{align}
Substituting the values of $\bSmp_i, l_i$, 
\begin{align}
\sum\limits_{x \in \mathcal{X}} \frac{p_{\cA} \left(I_i \middle| x \right)}{\bSmp_i}  &\leq \frac{\eps}{C_N } \left(\log^{\left(i\right)} k\right)^{\gamma} \exp\left( -C_N\frac{\left( \log^{\left(i\right)} k\right)^{\beta - \gamma}}{6\eps} \right) + 2\frac{\eps}{C_N }\left(\log^{\left(i\right)} k\right)^{\gamma-\beta}  \leq \frac{3}{C_N}\frac{\eps}{\left(\log^{\left(i\right)}k\right)^{\gamma}}. \nonumber 
\end{align}
The last inequality holds because $\beta = 2\gamma$ and $e^{-x} \le \frac{1}{x^2}$ for $x > 0$. Hence,
\begin{align} \label{eqn:sum_T-1}
	\sum\limits_{i=1}^{T-1} \left( \sum\limits_{x\in \mathcal{X}} \frac{p_{\cA} \left(I_i \middle| x \right)}{\bSmp_i} \right) \leq \frac{3\eps}{C_N} \sum\limits_{i=1}^{T-1} \frac{1}{\left(\log^{\left(i\right)}k\right)^{\gamma}} = \frac{3\eps}{C_N}  \sum\limits_{i=1}^{T-1} \frac{1}{\left(\log^{\left(T - i\right)}k\right)^{\gamma}}. 
\end{align}

Let $a_i = \log^{\left(T - i\right)}k$. Then $a_{i+1} = e^{a_i}$. Since $T = \log^* k$, we have $1 \le a_1\le e$. It can be shown that $\frac{a_{i+1}}{a_i} = \frac{e^{a_i}}{a_i} \ge e$. This implies 
\begin{align} \label{eqn:log_bound}
	\forall i, a_i = \log^{\left(T - i\right)}k \ge e^{i-1}.
\end{align}
Therefore,
\begin{align*}
	\sum\limits_{i=1}^{T-1} \frac{1}{\left(\log^{\left(T - i\right)}k\right)^{\gamma}} = \sum\limits_{i=1}^{T-1} \frac{1}{a_i^\gamma} \le  \sum\limits_{i=1}^{T-1} \frac{1}{e^{\gamma(i-1)}} \le 2.
\end{align*}

Plugging this in~\eqref{eqn:sum_T-1}, we can see $\exists$ constant $C_N > 36$, such that:
\begin{align*}
		\sum\limits_{i=1}^{T-1} \left( \sum\limits_{x\in \mathcal{X}} \frac{p_{\cA} \left(I_i \middle| x \right)}{\bSmp_i} \right)  \le \frac{\eps}{6}.
\end{align*}

For the $T^{th}$ interval, 
\begin{equation}
\sum\limits_{x \in \mathcal{X}} \frac{p_{\cA} \left(T \middle| x \right)}{\bSmp_T} \leq \frac{k}{\bSmp_T} \leq \frac{\eps}{6}.
\end{equation}
Adding the contributions from all the intervals gives us the desired bound.
\end{proof}

%% file: general-clipping.tex
We now bound the additional bias induced by the clipping step (Step~\ref{clipping_gen} of \textsc{GenCondExp}). 
\begin{lemma}(\textbf{Clipping Error bound})\label{lem:cliperr}
	Let $\esti$ be the estimate given by Algorithm~\ref{alg:genint} and $\estis$ be the entropy estimate without the clipping step in Algorithm~\ref{alg:esth_general} , then 
	\begin{equation} \label{eqn:lemma_clip}
	\left \lvert \expectation{\esti} -  \expectation{\estis}\right \rvert \le \frac{\eps}{3}. 
	\end{equation}
\end{lemma}
%\begin{lemma}
%	Algorithm~\ref{alg:general} with parameters defined by Equation~\ref{def:nj} and Equation~\ref{def:rj} satisfies the following:
%	\[
%		|\expectation{\widehat{H}} -  \expectation{\widehat{H}^*}| \le \epsilon/3
%	\]
%\end{lemma}
\begin{proof}
	
As before, $\bar{H}_i^*$ is the unclipped version of $\bar{H}_i$.  Hence we have

\begin{align}
	\left \lvert \expectation{\estis} - \expectation{\esti} \right \rvert & = \left \lvert \sum_{i = 1}^{T} p_{\cA}(I_i) \expectation{\bar{H}_i} - \sum_{i = 1}^{T} p_{\cA}(I_i) \expectation{\bar{H}_i^*}  \right \rvert  \nonumber \\
	& = \left \lvert \sum_{i = 1}^{T} p_{\cA}(I_i) (\expectation{\bar{H}_i} - \expectation{\bar{H}_i^*}) \right \rvert \le \sum_{i = 1}^{T} p_{\cA}(I_i)  \left \lvert  \expectation{\bar{H}_i} - \expectation{\bar{H}_i^*} \right \rvert. \label{eqn:sum_1}
\end{align}

Let's first bound each term $p_{\cA}(I_i)   \left \lvert \expectation{\bar{H}_i} - \expectation{\bar{H}_i^*}\right \rvert$ separately. Let $E_{X,i}^* = \log \left(\dfrac{\bSmp_i}{N_{X,i} + 1} \right) $ be the unclipped version of $E_{X,i}$ during each round when we are trying to estimate $H_i$. As we can see from the algorithm, $E_{X,i}^*$'s are independent and $\bar{H}_i^*$ is the empirical average of $E_{X,i}^*$'s in the same batch. Therefore,
\[
	\expectation{\bar{H}_i^*} = \expectation{E_{X,i}^*}.
\]
Similarly,
\[
	\expectation{\bar{H}_i} = \expectation{E_{X,i}}.
\]
Hence we have $p_{\cA}(I_i)  |\expectation{\bar{H}_i} - \expectation{\bar{H}_i^*}| = p_{\cA}(I_i)  |\expectation{E_{X,i}} - \expectation{E_{X,i}^*}|$. When $\dfrac{\bSmp_i}{N_{X,i} + 1} \ge \frac{1}{4h_i}$, we have $E_{X,i} = \max \left\{ \log \left(\dfrac{\bSmp_i}{N_{X,i} + 1} \right), \log \frac{1}{4h_i}\right\}  = E_{X,i}^*$ 

Next, consider the case when $\dfrac{\bSmp_i}{N_{x,i} + 1}  \in [\frac{\bSmp_i}{\bSmp_i+1}, \frac{1}{4h_i})$, which is $N_{X,i} \in (4h_i\bSmp_i-1, \bSmp_i]$. We divide the interval into $i -1$ intervals, which are
 \[
 L_1 = (4\bSmp_ih_2-1, \bSmp_i], L_2 = (4\bSmp_ih_3-1, 4\bSmp_ih_2 -1], ... , L_{i-1} = (4\bSmp_ih_i-1, 4\bSmp_ih_{i-1} -1]
 \]
The corresponding ranges of $\dfrac{\bSmp_i}{N_{X,i} + 1}$ are $ [\frac{\bSmp_i}{\bSmp_i+1}, \frac{1}{4h_2}),  [\frac{1}{4h_2}, \frac{1}{4h_3}), ... , [\frac{1}{4h_{i-1}}, \frac{1}{4h_i})$.

%Before we try to bound $p_{\cA}(I_i)  |\expectation{E_{X,i}} - \expectation{E_{X,i}^*}|$, it is necessary to understand the statistical distribution of $X$. Notice that suppose $X^* \sim p$, then $X$ is distributed according to the following distribution:
%
%\begin{align*}
%	\probof{X = x} &= \probof{X^* = x|\textsc{GenEstInt}(\bset{i}	,X^*) = I_i} \\ & = \frac{\probof{X^* = x, \textsc{GenEstInt}(\bset{i},X^*)  = I_i}}{\probof{\textsc{GenEstInt}(\bset{i},X^*) = I_i}} 
%	= \frac{\p(x) p_{\cA}(I_i|x) }{p_{\cA}(I_i) }
%\end{align*}

Since we are conditioning on $\textsc{GenEstInt}(\bset{i},X) = I_i$, $X$ here is distributed according to $p_{\cA} (X|I_i)$. Then we can rewrite the difference as:
\begin{align}
	&p_{\cA}(I_i)  |\expectation{E_{X,i}} - \expectation{E_{X,i}^*}| =  p_{\cA}(I_i) \expectation {E_{X,i} - E_{X,i}^*} \nonumber \\
 = &  p_{\cA}(I_i) \sum_{t = 1}^{i-1} \sum_{s \in L_t} \sum_{x} \probof{N_{X,i} = s |  X = x} p_{\cA} (x|I_i)(\log \frac{1}{4h_i} - \log \dfrac{\bSmp_i}{s + 1})  \nonumber  \\
 \le & \sum_{t = 1}^{i-1}  p_{\cA}(I_i) \sum_{s \in L_t} \sum_{x} \probof{N_{X,i} = s | X = x} \frac{\p(x) p_{\cA} (I_i|x)}{p_{\cA}(I_i)} (\log \frac{1}{4h_i} - \log \dfrac{1}{4h_{t}})  \nonumber  \\
 \le & \sum_{t = 1}^{i-1} p_{\cA}(I_i) \sum_{x}  \probof{N_{X,i} \in L_t | X = x} \frac{\p(x) p_{\cA} (I_i|x) }{p_{\cA}(I_i) } \beta \logjabk{t} \nonumber \\
 = & \sum_{t = 1}^{i-1} \sum_{x}  \probof{N_{X,i} \in L_t| X = x} p_{\cA} (I_i|x) \p(x)\beta \logjabk{t}.  \label{eqn:sum_2}
\end{align}

By Chernoff bound, we can get if $p(x)< 2h_{t+1}$,
\[
	\probof{N_{X,i} \in L_t | X= x} \le \exp \Paren{-\frac{\bSmp_i h_{t+1}}{3}}.
\]

If $p(x) > 2h_{t+1}$,

\[
		p_{\cA} (I_i|x) \le \exp \Paren{-\frac{\bSmp_i h_{t+1}}{6}}.
\]

Hence 
\[
	\probof{N_{X,i} \in L_t| X = x}  p_{\cA} (I_i|x) \le \max \{ \probof{N_{X,i} \in L_t | X= x}, p_{\cA} (I_i|x)  \} \le\exp (-\frac{\bSmp_i h_{t+1}}{6})
\]

Recall that $\bSmp_i = C_N\cdot\frac{\absz}{\eps(\logjabk{i})^{\gamma}}, h_i = \frac{(\logjabk{i-1})^{\beta}}{\absz}$. Plugging in we get 
%\begin{align}
%	\probof{N_{X,i} \in L_t| X = x} p_{\cA} (I_i|x) \le \exp (-\frac{C_N T \logjabk{t}^{\beta}}{6\eps \logjabk{i}^{
%	\gamma}})  \le \frac{\eps}{3T} \frac1{\logjabk{t-1}^{\logjabk{t}^3}}
%\end{align}

%\begin{align}
%	\probof{N_{X,i} \in L_t| X = x} p_{\cA} (I_i|x) \le \exp (-\frac{C_N  \logjabk{t}^{\beta}}{6\eps \logjabk{i}^{
%	\gamma}})  \le \frac{\eps}{6} \exp \left( -\frac{C_N }{6\eps} \left( \frac{ \logjabk{t}^{\beta}}{ \logjabk{i}^{
%	\gamma}}-1\right) \right).
%\end{align}

\begin{align}
\probof{N_{X,i} \in L_t| X = x} p_{\cA} (I_i|x) \le \exp (-\frac{C_N  \logjabk{t}^{\beta}}{6\eps \logjabk{i}^{
		\gamma}})  \le \frac{6\eps \logjabk{i}^{
		\gamma}}{C_N  \logjabk{t}^{\beta}} \le \frac{6\eps }{C_N  \logjabk{t}^{\gamma}}.
\end{align}

Plugging it into Equation~\eqref{eqn:sum_2},
\begin{align*}
	p_{\cA}(I_i)  |\expectation{E_{x,i}} - \expectation{E_{x,i}^*}|  \le \sum_{t =1}^{i - 1} \frac{6\eps }{C_N  \logjabk{t}^{\gamma}} \beta  \logjabk{t} = \sum_{t =1}^{i - 1} \frac{6\beta \eps }{C_N  \logjabk{t}^{\gamma-1}}. 
\end{align*}

By~\eqref{eqn:log_bound}, 
\[
	p_{\cA}(I_i)  |\expectation{E_{x,i}} - \expectation{E_{x,i}^*}| \le \frac{6\beta \eps }{C_N} \sum_{t =1}^{i - 1} e^{t+1-T} \le \frac{18\beta \eps }{C_N} e^{i-T}
\]
%\[ 
%p_{\cA}(I_i)  |\expectation{E_{x,i}} - \expectation{E_{x,i}^*}|  \le \sum_{t =1}^{i - 1} \frac{\eps}{6}  \exp \left( -\frac{C_N }{6\eps} \left( \frac{ \logjabk{t}^{\beta}}{ \logjabk{i}^{
%	\gamma}}-1\right) \right) \le \frac{\eps}{6}\left( \eps \frac{i-1}{\logjabk{i}^{
%	\gamma}}\right).
%\] 

Plugging this into~\eqref{eqn:sum_1}, and summing over $T$ intervals, we get:
\begin{align*}
	\left \lvert \expectation{\esti} -  \expectation{\estis}\right \rvert \le \frac{18\beta \eps }{C_N} \sum_{i =1}^{T}  e^{i-T} \le \frac{36\beta \eps }{C_N}.
\end{align*}

Hence,~\eqref{eqn:lemma_clip} is true with $C_N > 108\beta$.
%Summing this over $T-T'$ intervals and using \eqref{eq:Tdashbound}, the contribution from the $T-T'$ intervals is atmost $\frac{\eps}{6}$.
%
%For the remaining $T'$ intervals, recall that $\bSmp_i = C_N\cdot\frac{\absz \log^*\left(\log^*k\right)}{\eps(\logjabk{i})^{\gamma}}, h_i = \frac{(\logjabk{i-1})^{\beta}}{\absz}$
%
%\begin{align}
%	\probof{N_{X,i} \in L_t| X = x} p_{\cA} (I_i|x) \le \exp (-\frac{C_N  \logjabk{t}^{\beta}}{6\eps \logjabk{i}^{
%	\gamma}})  \le \frac{\eps}{6T'} \exp \left( -\frac{C_N T' }{6\eps} \left( \frac{ \logjabk{t}^{\beta}}{ \logjabk{i}^{
%	\gamma}}-1\right) \right).
%\end{align}
%
%Plugging it into Equation~\ref{eqn:sum_2}, we get 
%%\[
%%	p_{\cA}(I_i)  |\expectation{E_{x,i}} - \expectation{E_{x,i}^*}|  \le \sum_{t =1}^{i - 1} \frac{\eps}{3T} \frac{\beta \logjabk{t} }{\logjabk{t-1}^{\logjabk{t}^3}} \le \frac{\eps}{3T} 
%%\]
%
%\[ 
%p_{\cA}(I_i)  |\expectation{E_{x,i}} - \expectation{E_{x,i}^*}|  \le \sum_{t =1}^{i - 1} \frac{\eps}{6T'}  \exp \left( -\frac{C_N T'}{6\eps} \left( \frac{ \logjabk{t}^{\beta}}{ \logjabk{i}^{
%	\gamma}}-1\right) \right) \le \frac{\eps}{6T'}.
%\] 
%Summing over the $T'$ intervals, the contribution of the lat $T'$ intervals is $\frac{\eps}{6}$.
%Plugging it into Equation~\ref{eqn:sum_1}, we get 
%\[
%	|\expectation{\bar{H}^*} - \expectation{\bar{H}}|  \le \frac{\eps}{3}. %\qedhere
%\]
\end{proof}

%% file: general-variance.tex
In this section, we will derive a high probability bound on $\left \lvert \esti - \Expc{\esti} \right \rvert$. In particular, we will prove the following lemma:
\begin{lemma}(\textbf{Concentration bound}) \label{lem:genvar}
	Let $\esti$ be the estimate given by Algorithm~\ref{alg:genint} , then
	\begin{equation}
	\Pr\left( \left \lvert \esti - \Expc{\esti} \right \rvert  > \frac{\eps}{3} \right) \leq \frac{1}{3}.
	\end{equation}
\end{lemma}

\begin{proof}

\begin{align}
\left \lvert \esti - \Expc{\esti} \right \rvert &= \left \lvert \sum\limits_{i=1}^{T} \widehat{p}_{\cA} \left(I_i \right) \bar{H}_i - \sum\limits_{t=1}^{T} p_{\cA}\left(I_i\right) \Expc{\bar{H}_i} \right \rvert \nonumber \\
&=  \left \lvert \sum\limits_{i=1}^{T} \widehat{p}_{\cA}\left(I_i \right) \bar{H}_i - \sum\limits_{t=1}^{T}p_{\cA} \left(I_i\right)\bar{H}_i + \sum\limits_{t=1}^{T}p_{\cA}\left(I_i\right)\bar{H}_i -  \sum\limits_{t=1}^{T} p_{\cA}\left(I_i\right) \Expc{\bar{H}_i} \right \rvert \nonumber \\
&\leq \sum\limits_{i=1}^{T-1} \left \lvert \widehat{p}_{\cA}\left(I_i\right) - p_{\cA}\left(I_i\right) \right \rvert \left \lvert \bar{H}_i - \bar{H}_T \right \rvert  + \sum\limits_{i=1}^{T} p_{\cA}\left(I_i\right) \left \lvert \bar{H}_i - \Expc{\bar{H}_i} \right \rvert. \label{eq:gen-hpvar}
\end{align}
%By union bound, for the above sum to be less than $\eps/3$ with probability at least $2/3$, it is enough to show that each of the terms is greater than $\eps/(6T)$ with probability at most $1-1/6T$. We now consider the first term containing $T-1$ summands.

Next, we will bound each of the term seperately. For the first $T-1$ terms, note that, because of the clipping step, $\bar{H}_i \in \left[ \log \left( \dfrac{1}{4h_i} \right), \log \bSmp_i \right]$. Hence we have:
 \[
 	\left \lvert \bar{H}_i - \bar{H}_T \right \rvert  \le \log(4N_Th_i) \le (\beta+1) \log \frac{\logjabk{i-1}}{\eps}
 \]
The estimation of $\widehat{p}_{\cA}\left(I_i\right)$ requires $\vSmp_i$ independent executions of \textsc{GenEstProbInt}. Therefore, by Hoeffding's inequality (Lemma~\ref{lem:hoeff}), we have
\begin{align*}\label{eq:gen-fir-var}
\Pr\left( \left \lvert \widehat{p}_{\cA}\left(I_i\right) - p_{\cA}\left(I_i\right) \right \rvert > t_i \right) \leq 2\exp\left( -\vSmp_i t_i^2\right).
\end{align*}

Choosing $t_i = \dfrac{\eps}{C_T\Paren{\log\left(4\bSmp_T h_i \right)}^{5/4}}$ ($C_T \ge 30$ and constant) for $i = 1, \dots, T-1$ and the value of $\vSmp_i$ from Algorithm~\ref{alg:genint}, the right hand expression can be bounded by:
\begin{equation}\label{eq:varbd-fir}
	2\exp\left( -\vSmp_i t_i^2\right) \leq 2\exp \left( - \frac{C_R \Paren{\log\Paren{\logjabk{i-1}/\eps}}^{1/2} }{C_T^2(\beta + 1)^{5/2}} \right) \le \frac{2 C_T^4(\beta + 1)^{5}}{C_R^2\logjabk{i}}.
\end{equation}
the last inequality follows from $e^{-x} \le \frac{1}{x^2}$ for $x > 0$. Combing these for all $T-1$ intervals, let
\begin{align*}
	A = \{\sum\limits_{i=1}^{T-1} \left \lvert \widehat{p}_{\cA}\left(I_i\right) - p_{\cA}\left(I_i\right) \right \rvert \left \lvert \bar{H}_i - \bar{H}_T \right \rvert   \ge \sum\limits_{i=1}^{T-1} t_i \left \lvert \bar{H}_i - \bar{H}_T \right \rvert \}.
\end{align*}
Then by union bound and~\eqref{eqn:log_bound},
\begin{align*}
	\probof{A} \le \sum\limits_{i=1}^{T-1} \frac{2C_T^4(\beta + 1)^{5}}{C_R^2\logjabk{i}} \le \frac{2C_T^4(\beta + 1)^{5}}{C_R^2} \sum\limits_{i=1}^{T-1} \frac{1}{e^{i-1}} \le  \frac{4C_T^4(\beta + 1)^{5}}{C_R^2} . 
\end{align*}
Notice that
\begin{align*}
	\sum\limits_{i=1}^{T-1} t_i \left \lvert \bar{H}_i - \bar{H}_T \right \rvert  & \le  \sum\limits_{i=1}^{T-1} \dfrac{\eps}{C_T\Paren{\log\left(4\bSmp_T h_i \right)}^{5/4}}  \log(4N_Th_i) \le  \sum\limits_{i=1}^{T-1}  \frac{\eps }{C_T\Paren{(\beta+1)\logjabk{i}}^{1/4}} \\
	& \le  \sum\limits_{i=1}^{T-1} \frac{\eps }{C_Te^{\frac{i-1}{4}}} \le \frac{5\eps}{C_T} \le \frac{\eps}{6}.
\end{align*}
Hence we have for $C_R \ge 6C_T^2 (\beta + 1)^{5/2}$, we have:
\begin{align*}
	\probof{\sum\limits_{i=1}^{T-1} \left \lvert \widehat{p}_{\cA}\left(I_i\right) - p_{\cA}\left(I_i\right) \right \rvert \left \lvert \bar{H}_i - \bar{H}_T \right \rvert  \ge \frac{\eps}{6} } \le \probof{A}  \le \frac{1}{6}
\end{align*}
For the second term in~\eqref{eq:gen-hpvar}, %assuming $\vSmp_i p\left[I_i\right] > 8$ 
we use Lemma~\ref{lem:ranhoeff} where $p=p_{\cA}\left(I_i\right)$, $m=\vSmp_i$, $b = \log \bSmp_i$, $a = \log \frac{1}{h_i}$ to get 
\begin{align*} \label{eq:gen-sec-var}
\Pr\left( p_{\cA}\left(I_i\right) \left \lvert \bar{H}_i - \Expc{\bar{H}_i}  \right \rvert > c_i \right) &\leq  3\exp\left( \frac{-\vSmp_i c_i^2}{8p_{\cA}\left(I_i\right)\left( \log 4\bSmp_i h_i\right)^2}\right). \\
\end{align*}
Let $c_i = \dfrac{\eps}{C_c\Paren{\log\left(4\bSmp_T h_i \right)}^{1/4}}$ ($C_T \ge 30$ and constant) for $i = 1, \dots, T$. Using similar union bound argument as the first part, 
\begin{align*}
	\probof{\sum\limits_{i=1}^{T} p_{\cA}\left(I_i\right) \left \lvert \bar{H}_i - \Expc{\bar{H}_i} \right \rvert \ge \frac{\eps}{6}} \le \frac{1}{6}.
\end{align*}
Combining the two, we get
\begin{equation*} 
\Pr\left( \left \lvert\esti - \Expc{\esti} \right \rvert  > \frac{\eps}{3} \right) \leq \frac{1}{3}.
\end{equation*}
\end{proof}